\documentclass[aps,prx,twocolumn,showpacs]{revtex4-1}
\usepackage{bm}
\usepackage{mathrsfs}
\usepackage{amsmath}
\usepackage{amssymb}
\usepackage{graphicx}
\usepackage{amsfonts}
\usepackage{amsthm}
\usepackage{color}
\usepackage{booktabs}
\usepackage{tabularx}
\usepackage{dcolumn}
\usepackage{txfonts}
\usepackage{braket}
\usepackage[caption=false]{subfig}
\usepackage[T1]{fontenc}
\usepackage{multirow}

\DeclareMathOperator{\polylog}{polylog}

\newtheorem{lemma}{Lemma}

\newtheorem{theorem}{Theorem}

\newtheorem{conjecture}{Conjecture}

\begin{document}

\title{Near-Optimal Mode Scaling for Finite-Dimensional Boson Sampling via Lie-Algebraic Leakage Bounds}
\author{Chon-Fai Kam}
\email{dubussygauss@gmail.com}
\affiliation{Université Paris Cité \& University of Reunion, Paris, France}
\affiliation{Dipartimento di Fisica e Chimica Emilio Segrè, Università degli Studi di Palermo, Via Archirafi 36, I-90123 Palermo, Italy}

\author{En-Jui Kuo}
\email{kuoenjui@nycu.edu.tw}
\affiliation{Department of Electrophysics, National Yang Ming Chiao Tung University, Hsinchu, Taiwan, R.O.C.}
\affiliation{Center for Theoretical and Computational Physics, National Yang Ming Chiao Tung University, Hsinchu 300093, Taiwan}

\begin{abstract}
Boson sampling demonstrates quantum advantage through the interference of
indistinguishable particles, with output probabilities governed by matrix
permanents. Realizing it on deterministic, matter-based platforms---cold atoms,
trapped ions, superconducting cavities---requires encoding the bosonic modes in
finite-dimensional local Hilbert spaces, which introduces a leakage channel
absent in linear optics: multi-particle bunching beyond the local truncation
$d$. We develop a unified framework for non-interacting sampling on the
irreducible representations of compact Lie groups, in which the transition
amplitude is the immanant of a submatrix of the single-particle transition
matrix, recovering the permanent for the fully symmetric (bosonic) case. Within
this framework we bound the bunching leakage through a Dyson-series analysis:
decomposing the correlated many-body leakage operator into independent random
matrices and applying non-commutative concentration inequalities, we prove, in a
Gaussian model of the transition matrix, that its spectral norm concentrates at
$\tilde{O}(\sqrt{n})$ rather than the $O(n)$ worst-case estimate of prior
spin-based emulations; the passage to the physical Haar ensemble is reduced to a
single submatrix-comparison input, which we verify at leading order. Exact
numerics across local dimensions $d=2$--$5$ corroborate the concentration and
indicate that the bound is tight, the Haar-ensemble norm matching the closed
form $\sqrt{d(n-d+1)}$ to sub-percent accuracy. This tightens the required mode
number from
$m=\Omega(n^4)$
to the near-optimal $m=\tilde{\Omega}\!\left(n^{1+2/(d-1)}\right)$; for a spin-1
local representation ($d=3$) the overhead falls to $m=\tilde{\Omega}(n^2)$,
matching the collision-free threshold. The result is independent of particle
statistics and applies across finite-dimensional Lie-symmetric architectures,
quantifying the spatial resources needed to preserve sampling hardness on
deterministic hardware.
\end{abstract}

\maketitle

\section{Introduction}

Quantum sampling protocols have become a standard route to demonstrating
quantum computational advantage: they define tasks whose output distributions
are believed to resist efficient classical simulation, without requiring a
universal quantum computer~\cite{harrow2017quantum}. The paradigm is boson
sampling~\cite{Aaronson2011}, in which indistinguishable photons interfere in a
passive linear-optical network and the output probabilities are given by
permanents of submatrices of the network's transition matrix. Photonic
implementations of its Gaussian variant~\cite{hamilton2017gaussian} have reached
the quantum-advantage regime~\cite{zhong2020quantum,madsen2022quantum}. Computing the
permanent is \#P-hard~\cite{Valiant1979}, and under standard complexity-theoretic
conjectures this hardness extends to approximate sampling, placing an efficient
classical simulator in conflict with the non-collapse of the polynomial
hierarchy~\cite{aaronson2016bosonsampling,bouland2019complexity}.

The hardness argument is robust, but the photonic realization is not. Two
hardware bottlenecks limit its scalability: single photons must be generated
probabilistically, and the survival probability of an $n$-photon state decays
with the depth of the interferometer, so optical loss grows with photon
number~\cite{aaronson2016bosonsampling,garcia2019simulating}. Sufficiently lossy
instances are moreover efficiently simulable classically~\cite{renema2018efficient,oh2024classical,oszmaniec2018classical,clifford2018classical}. These are
limitations of the platform rather than of the complexity argument, and they
have motivated a shift toward deterministic, matter-based systems---cold atoms in
optical lattices, trapped ions, and superconducting resonator arrays---where
multi-particle states are prepared with high fidelity and particle number is
fixed by construction.

The shift comes at a structural cost. A photonic mode carries an unbounded local
occupation number, so each mode is an infinite-dimensional Fock space; a
matter-based mode instead carries a finite-dimensional local Hilbert space, set
by the dimension $d$ of its local representation. This truncation introduces a
failure mode absent in linear optics. When more than $d-1$ excitations bunch into
a single site, the dynamics leaves the physical subspace, and the correspondence
to the permanent breaks. Suppressing this bunching leakage requires diluting the
excitations across many modes, which turns the central practical question into
one of resource scaling: how many modes $m$ are needed, as a function of the
particle number $n$ and the local dimension $d$, to keep the leakage negligible
and the sampling distribution close to the ideal?

For the simplest case, spin-$1/2$ ($d=2$),
Peropadre et al.~\cite{peropadre2017equivalence} mapped boson sampling onto the
short-time evolution of a long-range $XY$ spin model and bounded the leakage
through a worst-case estimate of the relevant operator norm,
$\|QH_{\mathrm{Lie}}P\|\le O(n)$, giving $m=\Omega(n^4)$. They conjectured, on
the basis of small-size numerics ($n\le 6$), that for a Haar-random transition
matrix this norm should concentrate at $O(\sqrt{n})$, which would tighten the
requirement to $m\sim n^3$ and match the high-mode regime of photonic boson
sampling; a proof, however, was left open, and the conjecture has stood unproven
since. The obstruction is structural: the
many-body leakage operator has exponentially many entries, but they are
determined by only $O(m^2)$ single-particle parameters and are therefore
strongly correlated, so a naive random-matrix treatment that assumes independent
entries does not apply.

In this work we substantially advance this conjecture---proving it within a
Gaussian model of the transition matrix and reducing the Haar case to a single
submatrix-comparison input---and extend the framework to finite-dimensional
representations of arbitrary compact Lie groups. We first place boson, fermion,
and spin samplers in a single framework: non-interacting evolution on an
irreducible representation, with the transition amplitude given by the immanant~\cite{burgisser2000immanants}
of a submatrix of the single-particle matrix, of which the permanent (fully
symmetric, bosonic) is the \#P-hard case~\cite{Valiant1979}. This unifies the
spin-based~\cite{peropadre2017equivalence}, generalized-boson
\cite{kuo2022boson,kam2025beyond}, and fermionic~\cite{oszmaniec2022fermion} constructions as special cases of one
representation-theoretic structure. We then bound the bunching leakage through a
Dyson-series analysis. The key step is to decompose the correlated many-body
leakage operator into a sum of deterministic combinatorial operators weighted by
independent random variables, which makes non-commutative matrix concentration
inequalities~\cite{tropp2015introduction,tropp2012user} applicable; the residual deviation
between the Haar-random ensemble and its Gaussian approximation we reduce to a
single structural input, isolated below as Conjecture~\ref{conj:transfer} and
verified at leading order. This yields a concentration
of the leakage operator norm at $\tilde{O}(\sqrt{n})$\footnote{Throughout, $\tilde{O}$ and $\tilde{\Omega}$ suppress polylogarithmic factors: $f=\tilde{O}(g)$ means $f=O(g\,\polylog(n,m))$, with $\tilde{\Omega}$ defined dually.} in the Gaussian model,
replacing the $O(n)$ worst-case estimate; the corresponding Haar statement, which
is what Ref.~\cite{peropadre2017equivalence} conjectures, follows conditional on
that comparison input. The resulting mode requirement, stated as Theorem~\ref{thm:main} after the framework is set up, is
\begin{equation}
    m=\tilde{\Omega}\!\left(n^{\,1+\frac{2}{d-1}}\right),
\end{equation}
which for hard-core qubits ($d=2$) gives $m=\tilde{\Omega}(n^3)$, and for a
spin-$1$ local representation ($d=3$) falls to $m=\tilde{\Omega}(n^2)$---the
point at which the physical truncation boundary coincides with the collision-free
threshold, so that the mode budget is near-optimal there: any further reduction
reintroduces the very collisions the dilution is designed to suppress. The bound
is independent of particle statistics---the immanant framework places bosonic,
fermionic, and spin samplers on the same footing---and applies across
finite-dimensional Lie-symmetric architectures.

We emphasize the scope of the claim. We do not establish new hardness: the
\#P-hardness of the permanent and the average-case conjectures it relies on are
inherited from boson sampling~\cite{Aaronson2011,aaronson2016bosonsampling,bouland2026linear}. What
we provide is the resource cost of preserving that hardness on deterministic
finite-dimensional hardware---a quantitative mode-scaling law controlled by the
local dimension $d$. We also flag the one analytical gap explicitly: the leakage
bound is rigorous in the Gaussian model, and its transfer to the physical
Haar ensemble rests on a comparison step (Conjecture~\ref{conj:transfer}) that we
verify at leading order but do not establish in general for the structured,
deterministically sparse operator at hand. This evidence is strictly stronger than
that available for the original conjecture, which rested on small-size numerics
($n\le 6$): a complete proof in the Gaussian ensemble together with a
leading-order verification of the Haar correction, rather than numerical support
alone.

Finally, this scaling law presupposes that the single-particle transition matrix
can be generated in time independent of system size. On platforms restricted to
local interactions, Lieb-Robinson causality forces the scrambling time to grow
with the array diameter~\cite{Richerme2014}, and the extended evolution reopens
the leakage channel the dilution was meant to close. The mode-scaling law is
therefore operative only on architectures with fast, non-local connectivity,
which we use as the criterion for assessing candidate platforms.

The remainder of the paper is organized as follows.
Section~\ref{sec:framework} develops the generalized Lie-algebraic sampling
framework and the immanant structure of the amplitudes.
Section~\ref{sec:mainresult} contains the leakage analysis: the operator
decomposition, the matrix-concentration bound on $\|QH_{\mathrm{Lie}}P\|$, the
Weingarten decoupling, and the resulting mode-scaling law.
Section~\ref{sec:experiment} evaluates physical platforms through the combined
lens of finite-dimensional fidelity and non-local connectivity.
Section~\ref{sec:conclusion} summarizes and discusses verifiability under
residual interactions.

\section{Generalized Lie-Algebraic Sampling Framework}
\label{sec:framework}

The computational hardness of Aaronson--Arkhipov (AA) boson sampling arises from
the interference of non-interacting photons in a linear-optical network. On the
deterministic, matter-based platforms that motivate this work, what changes is
the local state space: a photonic mode carries an unbounded occupation number
and hence an infinite-dimensional Fock space, whereas a matter-based mode
carries a finite-dimensional local Hilbert space, set by the dimension $d$ of
its local representation. This section develops the framework that connects such
finite-dimensional systems to the permanent, and identifies when the resulting
sampling problem inherits the hardness of AA.

The finite local dimension is not an abstraction; it is produced by strong
interactions. Consider bosons on $m$ sites with a strong on-site repulsion, the
Bose--Hubbard model~\cite{childs2014bose}
\begin{equation}
    H_{\mathrm{BH}} = -J\sum_{\langle i,j\rangle} b_i^\dagger b_j
    + \frac{U}{2}\sum_i n_i(n_i-1),
    \qquad U \gg Jn .
    \label{eq:bose_hubbard}
\end{equation}
The repulsion costs an energy $\sim U$ for each excitation beyond the first on a
site, so when $U$ dominates the kinetic scale the low-energy manifold is the
truncated space in which each site holds at most $d-1$ excitations; states above
this cap sit at energies $\sim U$ and decouple from the slow dynamics. Within the
truncated manifold the bosonic operators no longer generate the full Fock tower
and instead close on a finite-dimensional algebra. In the hard-core limit $d=2$
the constraint $(b_i^\dagger)^2=0$ makes $b_i^\dagger, b_i$ obey the
$\mathfrak{su}(2)$ relations (the Matsubara--Matsuda map); a higher cap $d-1=2S$
gives the spin-$S$ representation, $d=2S+1$, and multi-component atoms with $d$
internal levels realize higher-rank $\mathfrak{su}(d)$
subalgebras~\cite{Zhang2014}. The same truncation arises dynamically in cavity
and circuit QED through the Kerr effect: a mode with a Kerr nonlinearity
$-\tfrac{K}{2}\,a^\dagger a^\dagger a\,a$ shifts the energy of the
$\lvert n\rangle$ Fock state away from its harmonic value, so that for $K$ large
compared to the drive bandwidth the cavity absorbs at most $d-1$ photons---a
photon blockade~\cite{leghtas2015confining,mirrahimi2014dynamically}. The
two-level limit ($d=2$) is the transmon; tuning $K$ or using selective drives
engineers $d=3$ and beyond~\cite{touzard2019gated}.

Either mechanism leaves the same effective dynamics: a network of $m$ sites, each
carrying a $d$-dimensional local space spanned by $\lvert 0\rangle,\dots,
\lvert d{-}1\rangle$, coupled by a hopping term that moves one excitation at a
time and conserves their total number,
\begin{equation}
    H_{\mathrm{Lie}} = \sum_{i,j=1}^{m}
    \left( A_{ji}\, E^{(j)}_{+}E^{(i)}_{-} + \mathrm{H.c.} \right),
    \label{eq:hlie}
\end{equation}
where $E^{(k)}_{\pm}$ are the raising and lowering generators of the local
algebra at site $k$, and the Hermitian matrix $A$ parametrizes the
single-particle transitions. The matrix $A$ plays the role of the unitary in
linear optics: it carries the amplitudes between input and output modes. For the
hardness analysis the relevant object is the single-particle sampling unitary
$U=e^{-iA\tau}$, which must be Haar-distributed on $U(m)$; the corresponding
Hermitian generator $A=i\log U/\tau$ is the physically synthesized coupling
ensemble. The leakage analysis of Section~\ref{sec:mainresult} is formulated for
coupling matrices with independent entries of variance $1/m$---the Gaussian
model of Eq.~\eqref{eq:decomposition} and the entries of a Haar unitary
(Appendix~\ref{app:weingarten})---neither of which is Hermitian. The Hermitian
constraint pairs the entries $A_{vu}$ and $A_{uv}^{*}$, which enter the leakage
operator through the distinct combinatorial operators $M_{vu}$ and $M_{uv}$;
the pairing nonetheless leaves the variance statistic of
Lemmas~\ref{lem:row}--\ref{lem:col} unchanged, the mixed contributions dropping
by the vanishing pseudo-variance $\mathbb E\,Z_{vu}^2=0$ of a circularly-symmetric
Gaussian (Lemma~\ref{lem:hermitian}); the numerical
comparison of Section~\ref{subsec:numerics} confirms that Hermitian and
independent-entry ensembles yield leakage norms that agree within statistical
error at every size tested. We therefore state the bounds for the analytically
convenient independent ensembles.
Equation~\eqref{eq:hlie} is non-interacting in the precise sense that it is
linear in the generators---a single element of the Lie algebra, not a product of
generators---so the evolution it generates stays within one irreducible
representation rather than scattering across the spectrum.

This last property is what lets us read off the structure of the transition
amplitudes. Because $H_{\mathrm{Lie}}$ is one body, the $N$-particle propagator
acts as the $N$-fold tensor power of the single-particle evolution
$U=e^{-i\,\pi_{\mathrm{def}}(H)t}$, the $M\times M$ matrix in the defining
representation ($M$ being the number of modes). Restricting this tensor power
to the physical symmetry sector and projecting onto an input and an output
occupation pattern $\mathbf m,\mathbf n$ gives the amplitude as an invariant
polynomial of the $N\times N$ submatrix $U[\mathbf m,\mathbf n]$ picked out by the
occupied modes,
\begin{equation}
    \mathcal{A}_{\mathbf m\to\mathbf n}
    \;\propto\; \mathrm{Imm}^{\lambda}\!\big(U[\mathbf m,\mathbf n]\big)
    \;=\; \sum_{\sigma\in S_N}\chi^{\lambda}(\sigma)
      \prod_{k=1}^{N} U_{j_k,\,i_{\sigma(k)}} ,
    \label{eq:immanant}
\end{equation}
where $\lambda$ is the Young diagram fixed by the exchange statistics of the
hardware, $\chi^\lambda$ is the corresponding character of the symmetric group
$S_N$, and the polynomial $\mathrm{Imm}^\lambda$ is the matrix \emph{immanant}.
The reduction of the many-body amplitude to Eq.~\eqref{eq:immanant}---the
Schur--Weyl decomposition, the role of the Young idempotent for mixed diagrams,
and the restriction to the unitary family---is carried out in
Appendix~\ref{app:tensor_factorization}; here we need only the resulting
correspondence between the shape of $\lambda$ and the type of the amplitude. A
single-row diagram, $\lambda=(N)$, has $\chi^\lambda\equiv 1$ and gives the
\emph{permanent}---the bosonic, fully symmetric case, and the \#P-hard object of
AA boson sampling~\cite{Valiant1979}. A single-column diagram, $\lambda=(1^N)$,
has $\chi^\lambda=\mathrm{sgn}$ and gives the \emph{determinant}, the fermionic
case, computable in polynomial time. Mixed shapes give intermediate immanants
that interpolate between these extremes. (Here ``interpolate'' refers to the
diagram shape, not to computational complexity, which is not monotone in
$\lambda$~\cite{burgisser2000immanants}.)

The diagram shape alone, however, does not settle the complexity, and it is worth
being precise about what else is required, because the natural intuition---``a
symmetric diagram, therefore hard''---is false. A permanent is \#P-hard only as
its size grows; a permanent of fixed dimension is computed in constant time.
Hardness needs two conditions together: the amplitude must lie on the permanent
side of the immanant easy/hard boundary, \emph{and} the rank $r$ of the algebra
must grow with the particle number $N$, so that the relevant permanent has
diverging dimension. The two are coupled, since in the dilute regime a fully
symmetric representation can grow only by adding modes---that is, by increasing
rank---which is exactly what drives the dimension of the permanent upward. Two
counterexamples fix the boundary. Free fermions have growing rank, yet their
antisymmetric diagram yields a determinant, absorbed by Gaussian elimination and
efficiently computable, so growing rank is not sufficient. A single large
spin-$S$ has a symmetric diagram but fixed rank ($\mathfrak{su}(2)$, $r=1$); its
amplitudes are classically computable $D$-functions whatever the size of the
diagram, so a symmetric diagram is not sufficient either. Only the bosonic
defining representation of $U(M)$, with both a symmetric diagram and a rank that
grows with the number of modes, places the amplitude on the \#P-hard permanent
and recovers AA. In the language of generalized coherent states, the
growing-rank condition is the statement that the representation possesses a
genuine many-body classical limit---a coadjoint orbit of diverging
dimension~\cite{kam2023coherent}.

We restrict the analysis that follows to this case. For the orthogonal and
symplectic families the commutant of the tensor power is not the symmetric group
but the Brauer algebra, and the amplitudes are Brauer-algebra invariants rather
than $S_N$-immanants; the shape/complexity correspondence must be re-derived
there, and we do not claim it. The remainder of the paper concerns the unitary
family in its defining representation---the permanent case---realized by the
truncated bosonic platforms of Eq.~\eqref{eq:hlie}.

The generators $E^{(k)}_{\pm}$ in Eq.~\eqref{eq:hlie} build the local states of
each site, and their finite-dimensional algebra fixes how those states are
normalized---a point we will need when bounding leakage. Take a single site with
lowest-weight state $\lvert 0\rangle$, defined by $E_-\lvert 0\rangle = 0$ and
$H\lvert 0\rangle = -\Lambda\lvert 0\rangle$, where the weight $\Lambda>0$ labels
the local irreducible representation. Repeated action of $E_+$ generates the
local tower $\lvert k\rangle \propto (E_+)^k\lvert 0\rangle$. Unlike the bosonic
case, the non-Abelian commutation relations make the norm of these states depend
on $k$ in a way set by the algebra: writing
$[E_-,E_+]=H$ and $[H,E_\pm]=\pm c\,E_\pm$ with $c$ fixed by the structure
constants, one finds
\begin{equation}
    \mathcal{W}(k,\Lambda)
    \;\equiv\; \langle 0\rvert (E_-)^k (E_+)^k \lvert 0\rangle
    \;=\; k!\prod_{p=1}^{k}\Big(\Lambda - c\,\tfrac{p-1}{2}\Big),
    \label{eq:weight}
\end{equation}
so that the normalized local basis state is
$\lvert k\rangle = \mathcal{W}(k,\Lambda)^{-1/2}\,(E_+)^k\lvert 0\rangle$.

The single function $\mathcal{W}$ interpolates between the statistics the
framework is meant to unify. For the bosonic (Heisenberg--Weyl) case $c=0$, and
\eqref{eq:weight} reduces to $\mathcal{W}(k)=k!\,\Lambda^k$, the standard bosonic
normalization up to a rescaling of the generators; the tower is unbounded, $d\to
\infty$. For the $\mathfrak{su}(2)$ representation of spin $S$ one has $c=2$ and
$\Lambda=2S$, giving $\mathcal{W}(k)=k!\prod_{p=1}^{k}(2S-p+1)$. The product vanishes once
$k>2S$: the weight $\mathcal{W}(k)=0$ for $k\ge d$ encodes the hard-core truncation
exactly, cutting the local space off at dimension $d=2S+1$. The vanishing of
$\mathcal{W}$ at $k=d$ is precisely the boundary that the dynamics must not
cross, and it is the bunching of more than $d-1$ excitations onto one site---the
event $\mathcal{W}=0$ would forbid---that constitutes the leakage we bound in
Section~\ref{sec:mainresult}.

A remark on where this leakage lives, since it must be well defined even on
platforms whose local space is fundamentally finite. The generators
$E_\pm^{(k)}$, the projectors $Q,P$, and the leakage operator
$M_{vu}=QE_+^{(v)}E_-^{(u)}P$ are defined on the \emph{parent} bosonic Fock
space---the untruncated space in which the state $\lvert d\rangle$ exists---and
$QH_{\mathrm{Lie}}P$ measures the deviation of the ideal bosonic hopping from the
truncated dynamics the device actually runs. This deviation is the same object
whether the truncation is emergent (a Bose--Hubbard or Kerr low-energy manifold,
where $\lvert d\rangle$ is a real but gapped state) or fundamental (a spin-$S$ or
hyperfine manifold, where the ladder terminates, $E_+\lvert d-1\rangle=0$): in
both cases it is the amplitude the ideal dynamics would carry across the
truncation boundary $\mathcal{W}=0$ but the device cannot. The leakage analysis
therefore applies uniformly to emergent- and exact-truncation platforms alike,
with no separate treatment required.

One caveat completes the setup. The truncation that produces
Eq.~\eqref{eq:hlie} is exact only as $U\to\infty$ (or $K\to\infty$). At finite
$U$, virtual excursions into the gapped sector generate superexchange corrections
of order $J^2/U$, of density-density ($n_in_j$) or Ising ($Z_iZ_j$) type. These
do not undermine the protocol. They commute with the total excitation
number---a consequence of the $U(1)$ symmetry that survives the projection---so
they keep the state within the fixed-$N$ sector where the permanent structure is
defined, and do not feed the bunching-leakage channel analyzed in
Section~\ref{sec:mainresult}. Each such term has strength $O(J^2/U)$ per bond; restricted to the fixed-$N$
sector the induced perturbation has operator norm $\lVert V\rVert = O(nJ^2/U)$, set
by the $O(n)$ occupied bonds, and enters the sampling error at order
$\tau\lVert V\rVert = O(nJ^2\tau/U)$---a tunable quantity, reduced by working at
larger $U/J$. The analysis of
Section~\ref{sec:mainresult} uses no spectral or ergodic property of $V$: it
enters only through its operator norm, as a coherent perturbation to
$e^{-iH_{\mathrm{Lie}}t}$ whose effect on the sampling distribution is bounded
directly. The hardness rests on the permanent structure of the
amplitude~\cite{Valiant1979,Aaronson2011}, not on any property of $V$.

The setup is now complete, and we can state the central result of the paper;
the remainder of the analysis is devoted to its proof.

\begin{theorem}[Mode-scaling law]
\label{thm:main}
Fix a local dimension $d\ge 2$ and let $n$ excitations evolve for a time
$\tau=O(1)$ under $H_{\mathrm{Lie}}$ of Eq.~\eqref{eq:hlie} on $m$ modes, with
the coupling matrix drawn from the Gaussian model
$A_{vu}=Z_{vu}/\sqrt{m}$, $Z_{vu}\sim\mathcal{CN}(0,1)$ independent. For every
$\delta>0$ there is a constant $c_d$ such that, writing $c_{d,\delta}=c_d\,\delta^{-2/(d-1)}$, with probability
$1-e^{-\Omega(n)}$ over the coupling matrix, the total-variation distance
between the output distribution of the finite-dimensional sampler and the
ideal boson-sampling distribution of Eq.~\eqref{eq:prob} is at most $\delta$
whenever
\begin{equation*}
    m \;\ge\; c_{d,\delta}\;
    \tilde{O}\!\left(n^{1+\frac{2}{d-1}}\right).
\end{equation*}
The same statement holds for the physical Haar ensemble, $A$ the Hermitian
generator of a Haar-random unitary, conditional on the single
Haar-to-Gaussian comparison input isolated as
Conjecture~\ref{conj:transfer}.
\end{theorem}

Table~\ref{tab:comparison} situates this against the prior bounds it
sharpens: the exponent improves from $m=\Omega(n^4)$ to
$m=\tilde\Omega(n^3)$ at $d=2$, resolving (in the Gaussian model, and
conditionally for Haar) the conjecture of
Ref.~\cite{peropadre2017equivalence}, and falls to the collision-free
threshold $m=\tilde\Omega(n^2)$ at $d=3$. The engine of the proof is a
concentration bound on the bunching-leakage operator norm
(Theorem~\ref{thm:leakage} below), whose strategy is sketched in
Fig.~\ref{fig:leakage}; the numerical verification of
Section~\ref{subsec:numerics} supports both the concentration and its
tightness.

\begin{table*}[t]
\centering
\caption{Mode-scaling bounds for constant total-variation error $\delta$ across
finite-dimensional sampling models, where $d$ is the local Hilbert-space dimension
(the highest-weight representation dimension of the local algebra; $d=2$ for
spin-$\tfrac{1}{2}$, $d=2S+1$ for spin $S$). For each model, $\delta$ is the
leakage error and the third column gives the mode count $m$ that keeps it from
growing. The Lie-algebraic result of this work sharpens the exponent to
$m=\tilde{\Omega}(n^{1+2/(d-1)})$, recovering the spin-$\tfrac{1}{2}$ ($d=2$) and
spin-$S$ ($d=2S+1$) rows as special cases with an $n^{1/2}$ improvement from the
concentration bound. This sharpening is unconditional in the Gaussian model of
the coupling matrix; its transfer to the physical Haar ensemble is conditional on
Conjecture~\ref{conj:transfer}.}
\label{tab:comparison}

\begin{ruledtabular}
\begin{tabular}{lccc}
Model & $\delta$ & Constant-error scaling & Ref. \\
\hline
Spin-$\tfrac{1}{2}$
  & $\dfrac{n^{2}}{\sqrt{m}}$
  & $m=\Omega(n^{4})$
  & \cite{peropadre2017equivalence} \\[1.2ex]
Generalized bosons
  & $\dfrac{n^{2}}{\sqrt{m}}$
  & $m=\Omega(n^{4})$
  & \cite{kuo2022boson,kam2025quantum} \\[1.2ex]
Spin $S$
  & $\dfrac{n^{S+3/2}}{m^{S}}$
  & $m=\Omega\!\big(n^{1+3/(2S)}\big)$
  & \cite{kam2025beyond} \\[1.2ex]
Lie-algebraic ($d$)
  & $\dfrac{n^{(d+1)/2}}{m^{(d-1)/2}}$
  & $m=\tilde{\Omega}\!\big(n^{1+2/(d-1)}\big)$
  & This work
\end{tabular}
\end{ruledtabular}
\end{table*}

\section{Bunching Leakage and the Mode-Scaling Bound}
\label{sec:mainresult}

Section~\ref{sec:framework} established that, in the collision-free sector, the
sampling amplitude is the permanent of a submatrix of $A$, and that the local
truncation forbids more than $d-1$ excitations per site---the boundary
$\mathcal{W}=0$ of Eq.~\eqref{eq:weight}. The protocol works only if the
dynamics stays clear of that boundary. This section makes the requirement
quantitative: we determine how many modes $m$ are needed, as a function of the
particle number $n$ and the local dimension $d$, to keep the bunching leakage
negligible and the sampled distribution close to the ideal permanent
distribution. We first locate the leakage in the output probabilities, which
gives a preliminary scaling from a counting argument; the remainder of the
section replaces that estimate by a rigorous operator-norm bound.

Consider $n$ excitations injected into distinct input modes,
$\lvert\Psi_{\mathrm{in}}\rangle = \Lambda^{-n/2}\prod_{k=1}^{n}
E^{(s_k)}_{+,\mathrm{in}}\lvert\varnothing\rangle$, evolving under the
symmetry-preserving linear map $E^{(s_k)}_{+,\mathrm{in}}\to
\sum_j U_{j,s_k}E^{(j)}_{+,\mathrm{out}}$. Expanding the evolved state in the
normalized output basis and reading off the coefficient of an output occupation
pattern $\mathbf m=(m_1,\dots,m_m)$ with $\sum_j m_j=n$ gives, using the local
weights $\mathcal{W}$ of Eq.~\eqref{eq:weight},
\begin{equation}
    P_{\mathfrak{g}}(\mathbf s\to\mathbf m)
    = \frac{\lvert\mathrm{Per}(U_{\mathbf m,\mathbf s})\rvert^2}
           {\prod_{j} m_j!}\;
      \prod_{j}\chi_{\mathfrak{g}}(m_j),
    \;
    \chi_{\mathfrak{g}}(m_j)\equiv
    \frac{\mathcal{W}(m_j,\Lambda)}{m_j!\,\Lambda^{m_j}} .
    \label{eq:prob}
\end{equation}
The factor $\lvert\mathrm{Per}(U_{\mathbf m,\mathbf s})\rvert^2/\prod_j m_j!$ is
exactly the ideal boson-sampling probability; the bunching penalty
$\chi_{\mathfrak{g}}(m_j)$ collects the finite-dimensional corrections and
encodes the deviation of the matter-based platform from a true boson sampler.

The penalty has a sharp interpretation in two regimes. For collision-free output,
$m_j\in\{0,1\}$, one has $\chi_{\mathfrak{g}}(0)=\chi_{\mathfrak{g}}(1)=1$ from
Eq.~\eqref{eq:weight}, and Eq.~\eqref{eq:prob} reduces to the exact permanent
distribution: in the collision-free sector the finite-dimensional system samples
identically to boson sampling, and the \#P-hardness is intact. Once a collision
occurs, $m_j\ge 2$, the penalty departs from unity, and for a finite local
dimension it vanishes outright at $m_j\ge d$, where $\mathcal{W}=0$ forbids
the configuration. Bunching of order $d$ is therefore the failure mode: it is the
point at which the matter-based sampler ceases to reproduce the permanent
distribution.

The resource cost of avoiding it follows from a counting estimate. Distributing
$n$ excitations uniformly over $m$ modes, the generalized birthday paradox gives
a probability of order $n^{d}/m^{d-1}$ for forming a $d$-fold collision, so the
amplitude for the leading single-bunch error scales as
$n^{d/2}/m^{(d-1)/2}$. Requiring this to vanish asymptotically gives a
preliminary mode requirement
\begin{equation}
    m = \tilde{\Omega}\!\left(n^{\,1+\frac{2}{d-1}}\right),
    \label{eq:prelim_scaling}
\end{equation}
which for hard-core qubits ($d=2$) reads $m=\tilde\Omega(n^3)$ and for spin-$1$
($d=3$) reads $m=\tilde\Omega(n^2)$. This estimate, however, assumes the leakage
is controlled by the collision probability alone. The actual error is governed by
the norm of the operator that drives population into the bunched sector, and
bounding that norm---not merely counting collisions---is what the rest of this
section establishes. As we will see, the naive counting happens to give the
correct exponent only because the operator norm concentrates at
$\tilde O(\sqrt n)$; the worst-case value $O(n)$, used in earlier
analyses~\cite{peropadre2017equivalence}, would instead force the far more
demanding $m=\Omega(n^4)$.

\subsection{Perturbative error framework}
\label{subsec:perturbative}

To turn the counting estimate of Eq.~\eqref{eq:prelim_scaling} into a rigorous
bound, we track the leakage at the level of the evolved state rather than the
output probabilities. The construction generalizes that of
Peropadre et al.~\cite{peropadre2017equivalence}, who showed that exact boson
sampling is reproduced by the short-time evolution of a long-range $XY$ spin
model: the photonic modes are replaced by spin excitations, and the
linear-optical unitary by the evolution operator on the spin Hilbert space, with
the interference-driven permanent amplitudes recovered in the dilute limit. This
was later adapted to generalized statistics~\cite{kuo2022boson}; here we carry it
to an arbitrary finite-dimensional local algebra, with $H_{\mathrm{Lie}}$ of
Eq.~\eqref{eq:hlie} as the evolution generator.

A feature of this construction, which we use throughout, is that the evolution
time is short. Reproducing the sampling distribution requires only that the
single-particle matrix $A$ be drawn from a sufficiently scrambled ensemble, not
that the many-body evolution form a unitary $2$-design; for a dense Haar-random
$A$ the required interference is reached in a time independent of system size,
$\tau=O(1)$. We carry $\tau$ explicitly through the error bound and return to its
physical realizability---which is not automatic on platforms with local
interactions---in Section~\ref{sec:experiment}.

The initial state $\lvert\phi(0)\rangle = \Lambda^{-n/2}\prod_{k=1}^{n}
E^{(s_k)}_{+,\mathrm{in}}\lvert\varnothing\rangle$ evolves under
$H_{\mathrm{Lie}}$, and we decompose the evolved state as
$\lvert\phi(t)\rangle = Q\lvert\phi(t)\rangle + \lvert\epsilon(t)\rangle$, where
$Q$ projects onto the \emph{truncation sector}
$\mathcal H_Q=\{\vec n\in\{0,\dots,d-1\}^m:\sum_k n_k=n\}$, in which every site
respects the local dimension, and $\lvert\epsilon(t)\rangle$ collects the leaked
configurations with at least one site at occupation $\ge d$. The strictly
collision-free subspace $n_k\le 1$, on which Eq.~\eqref{eq:prob} reduces to the
permanent distribution, is contained in $\mathcal H_Q$ and coincides with it only
at $d=2$; for $d\ge 3$ the leakage map resolves a bunch by taking the
over-occupied site from $d$ to $d-1$, so its image lies in $\mathcal H_Q$ rather
than in the collision-free subspace, consistent with the occupation constraints
$n_u=d-1$ and $n_v\le d-2$ used in Lemmas~\ref{lem:row}--\ref{lem:col}. The ideal
target is the dynamics confined to that sector,
\begin{equation}
    i\partial_t\lvert\psi(t)\rangle = QH_{\mathrm{Lie}}Q\lvert\psi(t)\rangle,
    \qquad \lvert\psi(0)\rangle=\lvert\phi(0)\rangle,
    \label{eq:ideal_eom}
\end{equation}
and the quantity to bound is the deviation
$\lvert\delta(t)\rangle = Q\lvert\phi(t)\rangle - \lvert\psi(t)\rangle$.

Differentiating $\lvert\delta(t)\rangle$ and using the two equations of motion
gives
\begin{equation}
    i\partial_t\lvert\delta(t)\rangle
    = QH_{\mathrm{Lie}}Q\lvert\delta(t)\rangle
    + QH_{\mathrm{Lie}}\lvert\epsilon(t)\rangle,
    \label{eq:delta_eom}
\end{equation}
whose solution is the Dyson integral
\begin{equation}
    \lvert\delta(t)\rangle
    = -i\int_0^t e^{-iQH_{\mathrm{Lie}}Q(t-s)}\,
      QH_{\mathrm{Lie}}\lvert\epsilon(s)\rangle\,ds .
    \label{eq:dyson}
\end{equation}
The operator $QH_{\mathrm{Lie}}$ acting on $\lvert\epsilon(s)\rangle$
annihilates every configuration with more than one bunched site, because
$H_{\mathrm{Lie}}$ relocates a single excitation at a time; only the single-bunch
component of $\lvert\epsilon(s)\rangle$ survives. We may therefore insert the
projector $P$ onto the single-bunch error subspace $\mathcal H_P$---one site at
occupation $d$, all others at occupation $\le d-1$---and bound the norm of
Eq.~\eqref{eq:dyson} by
\begin{equation}
    \lVert\delta(t)\rVert
    \le \int_0^t \lVert QH_{\mathrm{Lie}}P\rVert\,
        \lVert P\lvert\epsilon(s)\rangle\rVert \, ds ,
    \label{eq:delta_bound}
\end{equation}
where $\lVert O\rVert$ denotes the operator norm (largest singular value).

The two factors in Eq.~\eqref{eq:delta_bound} play distinct roles. The second,
$\lVert P\lvert\epsilon(s)\rangle\rVert$, is the instantaneous amplitude for
the state to occupy a single $d$-fold collision, and it is controlled by the
counting argument of the previous subsection. For a local dimension $d$, the
generalized birthday paradox gives
\begin{equation}
    \lVert P\lvert\epsilon(s)\rangle\rVert
    \le O\!\left(\frac{n^{d/2}}{m^{(d-1)/2}}\right),
    \label{eq:collision_amp}
\end{equation}
which recovers $O(n/\sqrt m)$ for the hard-core case $d=2$; the amplitude bound
follows from the spin-$S$ birthday paradox of Ref.~\cite{kam2025beyond}
(App.~A, Eq.~(9)), with $d=2S+1$. This amplitude bound is a high-probability
statement over the random coupling, not a deterministic one: the generalized
birthday paradox controls the collision weight up to an event of exponentially
small probability over the Haar (or Gaussian) ensemble, of the same order as the
failure probability of the leakage bound. The failure probability of
Theorem~\ref{thm:main} is therefore the union of the two, and remains
$e^{-\Omega(n)}$. Nor is the estimate circular:
$\lVert P\lvert\epsilon(s)\rangle\rVert$ is the single-bunch weight developed by
the collision-free reference dynamics $QH_{\mathrm{Lie}}Q$, bounded by the
birthday count independently of the leakage correction $\lvert\delta\rangle$,
which re-enters only at the next order through the Dyson integral
Eq.~\eqref{eq:dyson}. The first factor,
$\lVert QH_{\mathrm{Lie}}P\rVert$, is the operator norm of the map that drives
population from the bunched subspace back into the collision-free sector. This
factor, not the collision amplitude, is the bottleneck, and the difficulty of the
problem is concentrated entirely in it.

The difficulty is a structural correlation, not a shortage of estimates. For the
spin-$1/2$ case, Ref.~\cite{peropadre2017equivalence} bounded this norm by the
worst-case kinetic energy of the $n$ excitations, obtaining the deterministic
estimate $\lVert QH_{\mathrm{Lie}}P\rVert\le O(n)$. Substituting this and the
$d=2$ collision amplitude into Eq.~\eqref{eq:delta_bound} gives a total-variation
error of order $n^2/\sqrt m$, and requiring it to vanish forces $m=\Omega(n^4)$.
On the basis of small-size numerics ($n\le 6$), the same work conjectured that
for a Haar-random transition matrix the typical norm should instead concentrate
at $O(\sqrt n)$, which would tighten the requirement to $m\sim n^3$ and match the
high-mode regime of photonic boson sampling; a proof was left open. The
obstruction is that the many-body operator $QH_{\mathrm{Lie}}P$ lives in a space
of dimension $O(m^n)$, while all of its nonzero entries are determined by the
$O(m^2)$ parameters of the single-particle matrix $A$. Its entries are therefore
strongly correlated, and modeling them as independent random
variables---the naive application of random-matrix theory---violates this
 structure and yields inconsistent bounds. The next two subsections resolve this
by isolating the independent randomness explicitly, which restores the
applicability of matrix concentration inequalities and yields the conjectured
$\tilde O(\sqrt n)$.

\begin{figure*}[t]
  \centering
  \includegraphics[width=\textwidth]{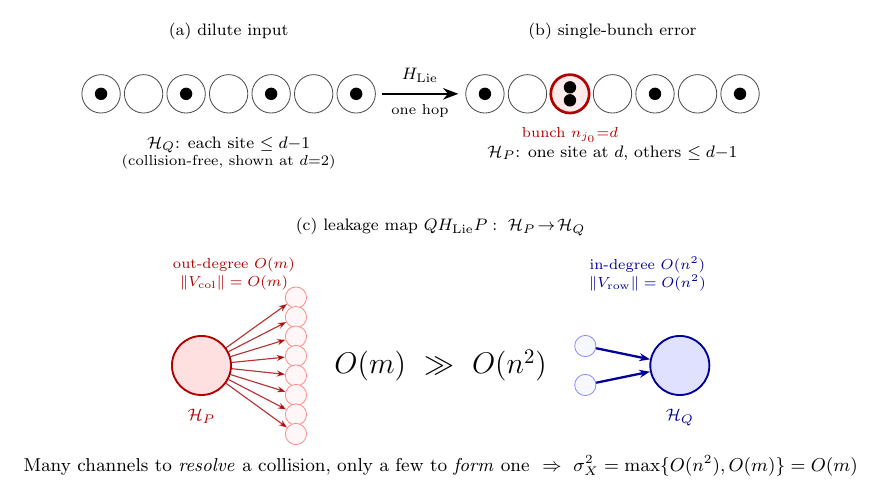}
  \caption{Sketch of the proof of Theorem~\ref{thm:leakage}: the bunching-leakage mechanism and the variance asymmetry behind it. \textbf{(a)} $n$ excitations injected into distinct
  modes populate the collision-free sector $\mathcal{H}_Q$ (occupation $\le d-1$;
  drawn at $d=2$). \textbf{(b)} A single action of $H_{\mathrm{Lie}}$ hops one
  excitation onto an occupied site, forming a $d$-fold bunch and entering the
  single-bunch error sector $\mathcal{H}_P$. \textbf{(c)} The leakage map
  $QH_{\mathrm{Lie}}P:\mathcal{H}_P\!\to\!\mathcal{H}_Q$ has a strongly asymmetric
  bipartite support: an $\mathcal{H}_P$ configuration scatters into $O(m)$ empty
  modes (out-degree $O(m)$, hence $\|V_{\mathrm{col}}\|=O(m)$), whereas an
  $\mathcal{H}_Q$ configuration is reached from only $O(n^2)$ collision-forming
  pairings (in-degree $O(n^2)$, hence $\|V_{\mathrm{row}}\|=O(n^2)$). The dilute
  dominance $O(m)\gg O(n^2)$ fixes $\sigma_X^2=O(m)$; the subsequent
  $1/\sqrt{m}$ rescaling turns this into the $\tilde{O}(\sqrt{n})$ leakage bound.}
  \label{fig:leakage}
\end{figure*}

\subsection{Operator decomposition and matrix variance}
\label{subsec:decomposition}

The correlation identified at the end of the previous subsection is resolved by
separating the randomness from the combinatorics. We rescale the leakage operator
as $X=\sqrt m\,QH_{\mathrm{Lie}}P$ and, applying the Gaussian approximation to the
single-particle entries $A_{vu}\to Z_{vu}/\sqrt m$ with $Z_{vu}\sim\mathcal{CN}(0,1)$
independent (this model is exact for the bound proved below; its transfer to the
Haar ensemble is discussed in Appendix~\ref{app:weingarten}, conditional on
Conjecture~\ref{conj:transfer} there), write $X$ as a linear
combination of deterministic many-body operators carrying independent scalar
weights:
\begin{equation}
    X = \sum_{u=1}^{m}\sum_{v\ne u} Z_{vu}\,M_{vu},
    \qquad
    M_{vu} = Q\,E^{(v)}_{+}E^{(u)}_{-}\,P .
    \label{eq:decomposition}
\end{equation}
Each $M_{vu}$ is a fixed matrix that encodes the combinatorics of annihilating an
excitation at site $u$ and creating one at site $v$, subject to the boundary
conditions imposed by $P$ and $Q$; all of the randomness sits in the independent
scalars $Z_{vu}$. With $X$ in this form, the non-commutative matrix Gaussian
inequality applies, and the bound on $\lVert X\rVert$ is controlled by the matrix
variance statistic
\begin{equation}
    \sigma_X^2 = \max\!\left\{
    \Big\lVert \sum_{u,v} M_{vu}M_{vu}^\dagger \Big\rVert,\;
    \Big\lVert \sum_{u,v} M_{vu}^\dagger M_{vu} \Big\rVert
    \right\}
    \equiv \max\{\lVert V_{\mathrm{row}}\rVert,\,\lVert V_{\mathrm{col}}\rVert\},
    \label{eq:variance_stat}
\end{equation}
using $\mathbb{E}\lvert Z_{vu}\rvert^2=1$. We bound the two variance operators in
turn. Their non-zero matrix elements are set by the local generators $E_\pm$,
whose amplitudes are bounded by a representation-dependent constant
$\gamma_{\mathfrak g}=O(1)$ since $d$ is a fixed hardware parameter
($d\ll n,m$); the scaling is therefore governed by the lattice combinatorics
rather than by the local operator weights. The bipartite support of the leakage
map and the resulting variance asymmetry $\|V_{\mathrm{col}}\|=O(m)\gg
\|V_{\mathrm{row}}\|=O(n^2)$ are summarized in Fig.~\ref{fig:leakage}.

\begin{lemma}[Row variance]
\label{lem:row}
The row-variance operator $V_{\mathrm{row}}=\sum_{u,v}M_{vu}M_{vu}^\dagger$ acts
within the truncation sector $\mathcal H_Q$, is diagonal in the Fock basis,
and satisfies $\lVert V_{\mathrm{row}}\rVert \le O(n^2)$.
\end{lemma}

\begin{proof}
For a normalized state $\lvert\Phi\rangle\in\mathcal H_Q$, expand
$M_{vu}M_{vu}^\dagger = Q\,E^{(v)}_{+}E^{(u)}_{-}\,P\,E^{(u)}_{+}E^{(v)}_{-}\,Q$.
Read right to left: $E^{(v)}_{-}$ requires site $v$ to hold at least one
excitation; $E^{(u)}_{+}$ then adds one to site $u$, and to enter the single-bunch
subspace through $P$ site $u$ must reach occupation $d$, so it must already hold
$d-1$ excitations (with $u\ne v$). The remaining operators reverse these moves and
restore the original configuration, so the operator is diagonal, scaling each
state by $\gamma_{u,v}^2\le\gamma_{\mathfrak g}^2$:
\begin{equation}
    V_{\mathrm{row}}\lvert\Phi\rangle
    = \sum_{u,v}\gamma_{u,v}^2\,
      \delta_{n_u,\,d-1}\,\Theta(n_v\ge 1)\,\lvert\Phi\rangle .
    \label{eq:row_diag}
\end{equation}
For any $\lvert\Phi\rangle\in\mathcal H_Q$, at most $\lfloor n/(d-1)\rfloor=O(n)$
sites satisfy $n_u=d-1$ and at most $n$ sites satisfy $n_v\ge 1$. The largest
eigenvalue is therefore bounded by
$\gamma_{\mathfrak g}^2\cdot O(n)\cdot O(n)=O(n^2)$.
\end{proof}

\begin{lemma}[Column variance]
\label{lem:col}
The column-variance operator $V_{\mathrm{col}}=\sum_{u,v}M_{vu}^\dagger M_{vu}$
acts within the single-bunch sector $\mathcal H_P$, is diagonal in the Fock basis,
and satisfies $\lVert V_{\mathrm{col}}\rVert = O(m)$.
\end{lemma}

\begin{proof}
For $\lvert\Psi\rangle\in\mathcal H_P$, exactly one site $u_0$ holds $d$
excitations and the rest satisfy $n_k\le d-1$. Expanding
$M_{vu}^\dagger M_{vu}=P\,E^{(u)}_{+}E^{(v)}_{-}\,Q\,E^{(v)}_{+}E^{(u)}_{-}\,P$,
the first action $E^{(u)}_{-}$ must resolve the truncation violation to survive
the projection $Q$, which forces $u=u_0$. The subsequent $E^{(v)}_{+}$ adds an
excitation at site $v$, which remains in the truncation sector $\mathcal H_Q$ only
if $v$ currently holds $n_v\le d-2$. The operator is again diagonal:
\begin{equation}
    V_{\mathrm{col}}\lvert\Psi\rangle
    = \sum_{v\ne u_0}\gamma_{u_0,v}^2\,\Theta(n_v\le d-2)\,\lvert\Psi\rangle .
    \label{eq:col_diag}
\end{equation}
The number of target sites with $n_v\le d-2$ is $m$ minus the at most
$\lfloor n/(d-1)\rfloor=O(n)$ sites at occupation $d-1$, leaving $m-O(n)$ terms.
In the scaling limit this gives $\lVert V_{\mathrm{col}}\rVert=O(m)$.
\end{proof}

The two bounds combine to $\sigma_X^2=\max\{O(n^2),\,O(m)\}$, and in the dilute
regime $m\gg n^2$ the column variance dominates,
\begin{equation}
    \sigma_X^2 = O(m).
    \label{eq:sigma_final}
\end{equation}
This asymmetry has a direct physical reading. Starting from a dilute,
collision-free state (occupation $\le 1$, the $d=2$ picture), there are only
$O(n^2)$ combinatorial pathways that bring two of the $n$ excitations together to
form a collision, but $O(m)$ empty modes into which a bunched excitation can
scatter to resolve one. The large mode space suppresses the leakage precisely
because resolving a collision has far more available channels than creating
one---the same dilution that makes the truncation sector accurate makes the
variance scale with $m$ rather than with the particle number.

\begin{lemma}[Hermitian reduction]
\label{lem:hermitian}
For the physical Hermitian ensemble $A_{vu}=A_{uv}^{*}$ with independent
upper-triangle entries $A_{vu}=Z_{vu}/\sqrt m$, $Z_{vu}\sim\mathcal{CN}(0,1)$
($u<v$), the rescaled leakage operator is
$X_H=\sum_{u<v}\big(Z_{vu}M_{vu}+Z_{vu}^{*}M_{uv}\big)$, and its second-moment
operators coincide with those of the independent model,
$\mathbb E\,X_HX_H^{\dagger}=V_{\mathrm{row}}$ and
$\mathbb E\,X_H^{\dagger}X_H=V_{\mathrm{col}}$. Hence
$\sigma_{X_H}^2=\sigma_X^2=O(m)$ and Theorem~\ref{thm:leakage} holds verbatim for
the Hermitian ensemble, with the same threshold.
\end{lemma}

\begin{proof}
Expanding
$\mathbb E\,X_HX_H^{\dagger}=\sum_{u<v}\sum_{u'<v'}\mathbb E\big[(Z_{vu}M_{vu}
+Z_{vu}^{*}M_{uv})(Z_{v'u'}^{*}M_{v'u'}^{\dagger}+Z_{v'u'}M_{u'v'}^{\dagger})\big]$,
distinct index pairs are independent and mean zero, so only $(u',v')=(u,v)$
survives. For each such pair the circular symmetry of the complex Gaussian gives
$\mathbb E|Z_{vu}|^2=1$ and $\mathbb E\,Z_{vu}^2=0$, so the mixed terms
$M_{vu}M_{uv}^{\dagger}$ and $M_{uv}M_{vu}^{\dagger}$ drop and
$\mathbb E\,X_HX_H^{\dagger}=\sum_{u<v}\big(M_{vu}M_{vu}^{\dagger}
+M_{uv}M_{uv}^{\dagger}\big)=\sum_{v\ne u}M_{vu}M_{vu}^{\dagger}=V_{\mathrm{row}}$,
the last equality because summing the ordered pair $\{(v,u),(u,v)\}$ over $u<v$
recovers the sum over all $v\ne u$. The identical computation gives
$\mathbb E\,X_H^{\dagger}X_H=V_{\mathrm{col}}$, whence
$\sigma_{X_H}^2=\max\{\lVert V_{\mathrm{row}}\rVert,\lVert V_{\mathrm{col}}\rVert\}
=\sigma_X^2=O(m)$ by Lemmas~\ref{lem:row}--\ref{lem:col}. Since $X_H$ is linear in
the Gaussians $Z_{vu}$, its Paulsen dilation is, after the same real/imaginary
splitting used for the independent model, a real Gaussian series of self-adjoint
matrices with this variance statistic and the same dilated dimension $D$; the
matrix Gaussian tail bound therefore yields the conclusion of
Theorem~\ref{thm:leakage} unchanged.
\end{proof}

\subsection{Spectral norm bound via concentration}
\label{subsec:concentration}

With the variance statistic of Eq.~\eqref{eq:sigma_final} in hand, we can state
and prove the central bound of the paper: the leakage operator norm concentrates
at $\tilde O(\sqrt n)$ rather than the worst-case $O(n)$. The bound is established
unconditionally in the Gaussian model; its statement for the physical Haar
ensemble, which is what the conjecture of Ref.~\cite{peropadre2017equivalence}
concerns, is conditional on the comparison input of
Conjecture~\ref{conj:transfer} below, and the framework extends to
arbitrary local dimension $d$.

\begin{theorem}[Leakage bound]
\label{thm:leakage}
For any fixed $\delta>0$ and $c>0$ there exists $c'>0$ such that, for the
Gaussian model $A_{vu}\to Z_{vu}/\sqrt m$, the leakage operator
norm obeys
\begin{equation}
    \Pr\!\left[\,\lVert QH_{\mathrm{Lie}}P\rVert
    \ge (1+\delta)\sqrt{n\,(\log m)^{1+c}}\,\right]
    \le \exp\!\big(-\Omega(n\,(\log m)^{c'})\big).
    \label{eq:theorem}
\end{equation}
In particular, in the regime $m=\tilde\Omega(n^3)$ one has
$\lVert QH_{\mathrm{Lie}}P\rVert \le \tilde O(\sqrt n)$ with probability
$1-1/\exp(\Omega(n))$. The same bound holds for the physical Haar ensemble
$A\in U(m)$ conditional on the Haar-to-Gaussian comparison of
Conjecture~\ref{conj:transfer}.
\end{theorem}

The rescaled operator $X=\sqrt m\,QH_{\mathrm{Lie}}P$ of
Eq.~\eqref{eq:decomposition} maps the single-bunch subspace $\mathcal H_P$ into
the collision-free subspace $\mathcal H_Q$, which have different dimensions
$k_2=\dim\mathcal H_P$ and $k_1=\dim\mathcal H_Q$, so $X\in\mathbb C^{k_1\times
k_2}$ is rectangular. The matrix concentration inequality we use is stated for
self-adjoint matrices, so we first pass to the Paulsen dilation
\begin{equation}
    \tilde X = \begin{pmatrix} 0 & X \\ X^\dagger & 0 \end{pmatrix}
    \in \mathbb C^{D\times D}, \qquad D = k_1 + k_2,
    \label{eq:dilation}
\end{equation}
which preserves everything we need. It is self-adjoint by inspection, and its
eigenvalues are the signed singular values of $X$: writing the singular-value
decomposition $X=\sum_i s_i\,u_iv_i^\dagger$, the vectors $(u_i,\pm v_i)$ are
eigenvectors of $\tilde X$ with eigenvalues $\pm s_i$, so that
$\lVert\tilde X\rVert=\max_i s_i=\lVert X\rVert$ exactly, with no loss in the
reduction. Its square is block-diagonal,
$\tilde X^2=\mathrm{diag}(XX^\dagger,\,X^\dagger X)$, so that
$\mathbb E\,\tilde X^2=\mathrm{diag}(V_{\mathrm{row}},V_{\mathrm{col}})$ in the
notation of Eq.~\eqref{eq:variance_stat}, and the variance statistic of the
dilation is exactly the one already computed in
Section~\ref{subsec:decomposition},
\begin{equation}
    \lVert\mathbb E\,\tilde X^2\rVert
    = \max\{\lVert V_{\mathrm{row}}\rVert,\,\lVert V_{\mathrm{col}}\rVert\}
    = \sigma_X^2 = O(m).
    \label{eq:dilation_variance}
\end{equation}
The only other quantity entering the bound is the dimension $D$, and it enters
only logarithmically. Both subspaces consist of configurations of at most $n$
excitations among $m$ modes, so each is contained in the corresponding Fock
space, $\dim\mathcal H_Q,\dim\mathcal H_P\le\binom{m+n-1}{n}$, giving
\begin{equation}
    D \le 2\binom{m+n-1}{n} \le 2\,m^n,
    \qquad \log D = O(n\log m).
    \label{eq:logD}
\end{equation}

We can now apply the concentration inequality. By Eq.~\eqref{eq:decomposition},
$\tilde X=\sum_{u}\sum_{v\ne u}Z_{vu}\,\tilde M_{vu}$ is a sum of independent,
mean-zero Gaussian matrices, where $\tilde M_{vu}$ is the dilation
of the deterministic operator $M_{vu}$. The coefficients are complex,
$Z_{vu}\sim\mathcal{CN}(0,1)$; since the tail bound is stated for a real Gaussian
series with self-adjoint summands, we split each into its real and imaginary parts,
$Z_{vu}=(g_{vu}+ig'_{vu})/\sqrt2$ with $g_{vu},g'_{vu}$ independent real standard
normals, which recasts $\tilde X$ as a real Gaussian series with twice as many
self-adjoint terms and the same variance statistic $\sigma_X^2$ of
Eq.~\eqref{eq:dilation_variance}; the threshold and scaling below are therefore
unaffected. For such a sum the matrix Gaussian tail
bound~\cite{tropp2015introduction} gives, for any $t\ge0$,
\begin{equation}
    \Pr\big(\lVert X\rVert\ge t\big)
    = \Pr\big(\lVert\tilde X\rVert\ge t\big)
    \le D\,\exp\!\left(-\frac{t^2}{2\,\sigma_X^2}\right),
    \label{eq:tail}
\end{equation}
where the equality is Eq.~\eqref{eq:dilation} and the variance is
Eq.~\eqref{eq:dilation_variance}. The prefactor $D$ is exponentially large, so
the threshold must be chosen large enough for the exponent to overcome it. Setting
$t^2/(2\sigma_X^2)=(1+c')\log D$ for a constant $c'>0$, the right-hand side
becomes $D\cdot D^{-(1+c')}=D^{-c'}=\exp(-\Omega(n\log m))$, an exponentially
small failure probability, at the threshold
\begin{equation}
    t = \sqrt{2(1+c')\,\sigma_X^2\,\log D}
      = O\!\big(\sqrt{m\cdot n\log m}\big),
    \label{eq:threshold}
\end{equation}
where we used $\sigma_X^2=O(m)$ and $\log D=O(n\log m)$. Inverting the rescaling
$X=\sqrt m\,QH_{\mathrm{Lie}}P$ divides the threshold by $\sqrt m$, cancelling the
spatial-volume factor that the dilute scaling introduced,
\begin{equation}
    \lVert QH_{\mathrm{Lie}}P\rVert_{\mathrm{Gauss}}
    = \frac{\lVert X\rVert}{\sqrt m}
    \le \frac{O(\sqrt{m\cdot n\log m})}{\sqrt m}
    = O\!\big(\sqrt{n\log m}\big),
    \label{eq:gauss_norm}
\end{equation}
which is the $\tilde O(\sqrt n)$ of the theorem, holding with probability
$1-\exp(-\Omega(n\log m))$ over the Gaussian ensemble.

It is worth pausing on where the improvement over the worst-case bound comes
from, since the whole result turns on it. The factor $\sqrt n$ originates
entirely in the entropic dimension $\log D=O(n\log m)$ of the Fock space, while
the dominant $O(m)$ variance is cancelled when the rescaling is inverted. The
worst-case estimate $O(n)$ of Ref.~\cite{peropadre2017equivalence} amounts to
bounding $\lVert X\rVert$ by its largest possible value; the concentration
inequality replaces this by its typical value, and the gap between the two is
precisely the gap between $n$ and $\sqrt n$---equivalently, between
$m=\Omega(n^4)$ and $m=\tilde\Omega(n^3)$.

This bound was established for the independent Gaussian model
$A_{vu}\to Z_{vu}/\sqrt m$, whereas the physical transition matrix is Haar-random.
The two are not identical: while the sub-blocks of a Haar unitary converge to
independent complex Gaussians as $m\to\infty$, the leakage operator contracts up
to $k\sim n\log m$ matrix elements at once, and the constraint $U^\dagger U=I$
couples them, so the passage to the Haar ensemble requires a comparison of the two
operator norms. For unstructured submatrices such a comparison is the standard one
underlying boson sampling~\cite{aaronson2016bosonsampling,leverrier2018p}; for the
structured, projector-constrained operator here we reduce it to a single
structural input, formulated as Conjecture~\ref{conj:transfer} in
Appendix~\ref{app:weingarten} and verified there at leading order. Granting that
input, the comparison takes the form of the following statement.

\begin{lemma}[Haar-to-Gaussian comparison; conditional on Conjecture~\ref{conj:transfer}]
\label{lem:weingarten}
Let $\mathcal A_{\mathrm{Haar}}$ and $\mathcal A_{\mathrm{Gauss}}$ be the leakage
operator generated by the exact Haar ensemble and by the independent Gaussian
approximation, respectively. If Conjecture~\ref{conj:transfer} holds, then in the
dilute limit their operator norms satisfy
\begin{equation}
    \lVert \mathcal A_{\mathrm{Haar}}\rVert
    \le \lVert \mathcal A_{\mathrm{Gauss}}\rVert
        \left(1 + \tilde O\!\left(\tfrac{n^4}{m^2}\right)\right).
    \label{eq:weingarten}
\end{equation}
\end{lemma}

The route to such a comparison is visible in the moments. The operator norm is
reached through the high moments $\mathbb E\,\mathrm{Tr}\big((XX^\dagger)^k\big)$
at order $k\sim n\log m$, each a Haar integral of degree $2k$ in the entries of
$U$. The Weingarten formula evaluates such an integral as a sum over pairs of
permutations $\sigma,\tau\in S_k$ weighted by $\mathrm{Wg}(\sigma^{-1}\tau,m)$; the
diagonal terms $\sigma=\tau$ reproduce the Gaussian (Wick) moment, and every
genuine unitary correlation sits in the off-diagonal terms. Conjecture~\ref{conj:transfer}
asserts that, when these terms are organized by genus, the projector decorations
$P,Q$ preserve the topological power counting, so that a contribution at distance
$\ell$ and genus $g$ is suppressed by $m^{-2(\ell-g)}$ relative to the diagonal and
the off-diagonal sum is controlled by its leading $\ell=2$, $g=1$ term,
$\tilde O(n^4/m^2)$. We verify this at leading order in
Appendix~\ref{app:weingarten}, but emphasize that the general statement is not
established: the projector constraints make the index-loop weights heterogeneous
(some loops are pinned to the $O(n)$ occupied sites, others range over the $O(m)$
empty modes, as already reflected in the variance asymmetry
$\lVert V_{\mathrm{row}}\rVert=O(n^2)$ versus $\lVert V_{\mathrm{col}}\rVert=O(m)$
of Lemmas~\ref{lem:row}--\ref{lem:col}), and whether the uniform genus ordering
survives this heterogeneity for all $\ell$ is open.

We therefore record the logical status precisely. The Gaussian-model bound,
Theorem~\ref{thm:leakage} with $A_{vu}\to Z_{vu}/\sqrt m$, is unconditional. Its
transfer to the physical Haar ensemble holds in the operating regime
$m=\tilde\Omega(n^3)$ conditional on Conjecture~\ref{conj:transfer}; this is the
sole analytical gap between the present analysis and a full resolution of the
conjecture of Ref.~\cite{peropadre2017equivalence}, the resulting
evidence---a complete proof in the Gaussian ensemble together with a leading-order
verification of the Haar transfer---being strictly stronger than the small-size
numerics ($n\le 6$) that originally motivated the conjecture.

\subsection{The mode-scaling law}
\label{subsec:scaling}

We now combine the leakage bound of Theorem~\ref{thm:leakage} with the collision
amplitude of Eq.~\eqref{eq:collision_amp} to obtain the mode count required for
faithful emulation. Within the perturbative framework of
Section~\ref{subsec:perturbative}, the total variation distance between the
finite-dimensional sampling distribution and the ideal one is bounded, in the
occupation basis, by the state-vector deviation $\lVert\delta(\tau)\rVert$ of
Eq.~\eqref{eq:delta_bound} (the trace-distance factor absorbed into the $O(1)$
prefactor),
\begin{equation}
    \varepsilon \;\le\; \tau\,\lVert QH_{\mathrm{Lie}}P\rVert\,
    \max_{t}\lVert P\lvert\phi(t)\rangle\rVert,
    \label{eq:tvd_product}
\end{equation}
the product of the evolution time $\tau$, the leakage operator norm, and the peak
collision amplitude. Here $\lVert QH_{\mathrm{Lie}}P\rVert$ is the operator-norm
bound of Theorem~\ref{thm:leakage} and
$\max_t\lVert P\lvert\phi(t)\rangle\rVert$ is the peak collision amplitude of
Eq.~\eqref{eq:collision_amp}---the same quantity, since $P$ annihilates
$\mathcal H_Q$ and hence $P\lvert\phi(t)\rangle=P\lvert\epsilon(t)\rangle$; the two factors share the single-bunch subspace
$\mathcal H_P$ but are distinct quantities---the former a spectral norm from
matrix concentration, the latter a state amplitude from the birthday-paradox
counting---and combine multiplicatively in the Dyson bound
\eqref{eq:delta_bound}. Substituting the operator-norm bound
$\lVert QH_{\mathrm{Lie}}P\rVert=\tilde O(\sqrt n)$ from
Theorem~\ref{thm:leakage}, the collision amplitude
$\lVert P\lvert\phi(t)\rangle\rVert = O\big(n^{d/2}/m^{(d-1)/2}\big)$ from
Eq.~\eqref{eq:collision_amp}, and the constant evolution time $\tau=O(1)$
discussed in Section~\ref{subsec:perturbative}, gives
\begin{equation}
    \varepsilon \;\le\; O(1)\cdot \tilde O(\sqrt n)\cdot
    O\!\left(\frac{n^{d/2}}{m^{(d-1)/2}}\right)
    = \tilde O\!\left(\frac{n^{(d+1)/2}}{m^{(d-1)/2}}\right).
    \label{eq:tvd_scaling}
\end{equation}
Requiring the error to vanish asymptotically, $\varepsilon\to 0$, fixes the
mode-scaling law: the denominator must outgrow the numerator, which holds when
\begin{equation}
    m = \tilde\Omega\!\left(n^{\,1+\frac{2}{d-1}}\right).
    \label{eq:scaling_law}
\end{equation}
This establishes Theorem~\ref{thm:main}. It states the spatial overhead---the
number of modes per excitation---needed to keep a finite-dimensional,
matter-based emulator within the regime where its output distribution coincides
with that of boson sampling, and hence inherits the same hardness.

The exponent is controlled entirely by the local dimension $d$, and two cases are
of direct experimental relevance. For hard-core qubits, $d=2$, the law gives
\begin{equation}
    m = \tilde\Omega(n^3),
    \label{eq:d2}
\end{equation}
which is the bound conjectured by Ref.~\cite{peropadre2017equivalence}, proved
here unconditionally in the Gaussian model and, for the physical Haar ensemble,
conditional on Conjecture~\ref{conj:transfer}; it improves on the previously established $m=\Omega(n^4)$ by the factor of
$n$ that separates the worst-case operator norm $O(n)$ from the typical norm
$\tilde O(\sqrt n)$. For a spin-$1$ local representation, $d=3$, the exponent
drops to
\begin{equation}
    m = \tilde\Omega(n^2),
    \label{eq:d3}
\end{equation}
which coincides with the collision-free threshold $m\gg n^2$ of the dilute limit
itself. At $d=3$ the physical truncation boundary and the dilution requirement
align: no spatial overhead beyond what dilution already demands is needed to
suppress bunching leakage. Higher local dimensions $d\ge 4$ reduce the exponent
further toward the linear algorithmic baseline, but $d=3$ is the point at which
the finite-dimensional constraint ceases to be the binding one, which makes it
the most resource-efficient setting for a finite-level demonstration.

The status of this optimum deserves comment, since it is tied to the hardness
conjecture the emulation inherits. The target reproduced by the finite-dimensional
sampler is the collision-free boson-sampling distribution, whose classical
hardness is established only in the dilute regime $m=\Omega(n^2)$
\cite{Aaronson2011,aaronson2016bosonsampling,bouland2026linear}. Within that
frame---the one in which the conjecture of Ref.~\cite{peropadre2017equivalence}
was posed, before the collisional regime had been studied---the collision-free
floor $n^2$ is itself a lower bound on the mode count, and $d=3$ is optimal: it is
the local dimension at which the truncation floor $n^{1+2/(d-1)}$ first meets
$n^2$, so that for $d\ge4$ the truncation ceases to bind while the collision-free
requirement does not. The $\tilde O(\sqrt n)$ concentration of
Theorem~\ref{thm:leakage} is what places this crossover at $d=3$ rather than at
the $d=4$ that the worst-case norm $O(n)$ would give. Whether higher $d$ can be
exploited depends on the hardness of boson sampling in the collisional regime
$m=o(n^2)$, an open question~\cite{grier2022bipartite,go2026shallow} that
postdates Ref.~\cite{peropadre2017equivalence} and lies outside its scope: should
collisional sampling prove hard, the leakage-only requirement
$m=\tilde\Omega(n^{1+2/(d-1)})$ would govern and approach the quasilinear regime as
$d$ grows, so that the higher-spin construction of Ref.~\cite{kam2025beyond}---whose
bound uses no worst-case operator norm---would become the operative one; absent
that, $d=3$ remains optimal.

Two remarks delimit the result. First, the law \eqref{eq:scaling_law} concerns
the \emph{spatial} resource: it counts modes, taking the evolution time as
$\tau=O(1)$. This is justified when the single-particle transition matrix $A$ can
be generated in time independent of system size, which requires the hardware to
support fast, effectively non-local coupling; on platforms restricted to local
interactions the evolution time grows with the array, and the bound must be
re-read with $\tau$ restored. We take this up in
Section~\ref{sec:experiment}. Second, the law inherits the regime of validity of
Theorem~\ref{thm:leakage}: the concentration bound and the Haar-to-Gaussian
transfer hold in the dilute regime, and the scaling \eqref{eq:scaling_law} is
stated within it. For $d=2$ the required $m=\tilde\Omega(n^3)$ coincides with that
regime, so the law is self-consistent at the boundary it predicts.

\subsection{Tightness and the lower bound}
\label{subsec:tightness}

The bound of Theorem~\ref{thm:leakage} is an upper bound:
$\lVert QH_{\mathrm{Lie}}P\rVert \le \tilde O(\sqrt n)$ with high probability.
Whether it is tight---whether a matching lower bound
$\lVert QH_{\mathrm{Lie}}P\rVert \ge \Omega(\sqrt n)$ holds---determines whether
the mode-scaling law \eqref{eq:scaling_law} is optimal or merely sufficient. We
have not established the lower bound, and we explain here why it lies beyond the
methods used above, and what a proof would require.

The obstruction is that the leakage operator resists the techniques that usually
deliver spectral lower bounds. For unstructured random ensembles the typical
largest singular value is pinned by universality results---the Marchenko--Pastur
edge, Tracy--Widom fluctuations---which assume independent or weakly dependent
entries. The leakage operator has neither: its entries are determined by the
$O(m^2)$ parameters of $A$ and are strongly correlated, and the projectors $P,Q$
impose severe combinatorial constraints that make it sparse and rank-deficient.
The concentration inequality we used is therefore one-sided by nature; it bounds
the norm from above through the variance statistic, but the same structure that
makes the variance small offers no lower bound on the extreme singular value.

A lower bound would have to come from the geometry of the operator rather than
from a global average. The natural object is the bipartite support graph of
$QH_{\mathrm{Lie}}P$ between the collision-free sector $\mathcal H_Q$ and the
single-bunch sector $\mathcal H_P$, whose connectivity is dictated by the
Lie-algebraic hopping. This graph is strongly asymmetric. From a vertex in
$\mathcal H_P$---a configuration with one bunched site---the out-degree is $O(m)$,
reflecting the empty modes into which the bunched excitation can scatter; from a
vertex in $\mathcal H_Q$ the in-degree is only $O(n^2)$, the number of pairings
that can form a collision from the dilute gas. This asymmetry concentrates the
operator's weight on a sparse set of high-degree hubs---star-graph motifs
embedded in the exponentially large Fock space. A lower bound on the spectral
radius could plausibly be extracted by constructing trial vectors localized on
these hubs and evaluating the Rayleigh quotient, or by a high-moment trace method
sensitive to the hub structure; we leave this to future work.

The same gap is the origin of the polylogarithmic factor in our upper bound. The
$\tilde O(\sqrt n)$ is reached through the moment method at order
$k\sim n\log m$, forced by the entropic dimension of the Fock space; the
$\log$ factors are an artifact of accessing the norm through high moments rather
than through the spectral edge directly. A method that controlled the edge
without passing through the $k$-th moment---an edge-universality or local-law
statement adapted to the sparse, correlated structure---would remove the
polylog and decide the lower bound at once. Both the optimality of the exponent
and the removal of the polylog therefore reduce to the same open problem: a
direct spectral analysis of structured, Fock-space leakage operators. Until it
is resolved, the result stands as a near-optimal upper bound: optimal up to a
polylogarithmic factor in the mode count, and sufficient to place
finite-dimensional emulation within the hardness regime.

The numerical results of the next subsection bear directly on this question:
the measured Haar-ensemble norm matches the closed form $\sqrt{d(n-d+1)}$ to
sub-percent accuracy across $d=2$--$5$, indicating that the hub Rayleigh
quotient survives the phase averaging and that
$\lVert QH_{\mathrm{Lie}}P\rVert=\Theta(\sqrt n)$ at fixed $d$; a variational
proof organized around the coherent $(n-d+1)$-fold donor sum, with the factor
$d$ supplied by the unitarity-pinned column weight identified there, is the
natural target.

\subsection{Numerical verification and a sharp-constant conjecture}
\label{subsec:numerics}

The bounds above can be tested directly at moderate sizes by exact
construction of the leakage operator. The single-bunch sector $\mathcal{H}_P$
and the reached portion of the truncation sector $\mathcal{H}_Q$ are
enumerated exactly, and $X=QH_{\mathrm{Lie}}P$ is assembled as a sparse matrix
whose combinatorial structure (the operators $M_{vu}$) is built once per
$(n,m,d)$ and refilled per realization of the coupling matrix; the
construction is validated against a brute-force diagonalization of the full
fixed-$N$ Fock-space Hamiltonian projected onto the $Q$ and $P$ sectors, the
sparse and dense norms agreeing to ten significant figures at all six test points
(Table~\ref{tab:validation}); the largest singular value is obtained by dense SVD
below dimension $300$ and by sparse SVD above it. The complete parameter grid,
with sector dimensions and dilution status, is given in Table~\ref{tab:grid}. We compare four coupling ensembles,
all normalized to entry variance $1/m$: the independent Gaussian model of
Eq.~\eqref{eq:decomposition}; its Hermitian counterpart, with the upper
triangle drawn independently; the entries of a Haar unitary, the ensemble of
Conjecture~\ref{conj:transfer}; and the Hermitian generator $i\log U$ of a
Haar unitary, the physically synthesized coupling. All runs use deterministic
per-point seeding and are reproducible from a single master seed. The largest-$n$
points of the $d=2$ and $d=3$ $n$-scans are sub-dilute ($m<2n^2$, i.e.\ from
$n=5$ onward; Table~\ref{tab:grid}) and appear as open markers in
Fig.~\ref{fig:numerics}, lying below the dilute-limit curve in the direction that
strengthens the bound.

Three findings result, summarized in Fig.~\ref{fig:numerics}. First, the
scaling of Theorem~\ref{thm:leakage} is confirmed and sharpened. Across
$3\le n\le 8$ at $d=2$ the fitted log-log slope of the median norm against $n$
is $0.55$ for the Haar ensemble and $0.55$--$0.68$ across all four ensembles,
decisively excluding the worst-case linear scaling. (A fitted slope near $0.55$
is itself consistent with the $\sqrt n$ law rather than in tension with it: the
effective log-log slope of the conjectured form $\sqrt{2(n-1)}$ over this same
window $3\le n\le8$ is $\approx0.56$, so the small excess above $1/2$ reflects the
finite-$n$ $-d+1$ offset, not a genuinely super-square-root exponent.) More
precisely, the
Haar-ensemble medians---whose realization-to-realization fluctuations are of
order $10^{-3}$, itself a signature of the unitary constraint discussed
below---are reproduced, in the dilute regime, by the closed form
\begin{equation}
    \lVert QH_{\mathrm{Lie}}P\rVert_{\mathrm{Haar}}
    \;=\; \sqrt{d\,(n-d+1)}\;\bigl(1+o(1)\bigr),
    \label{eq:sharp_constant}
\end{equation}
verified at $d=2,3,4,5$: the deviation is sub-percent throughout the dilute
regime---reaching $+0.03\%$ at $(d,n,m)=(2,3,384)$ and $4\times10^{-4}$
relative at $(3,5,30)$---and at $d\ge3$ remains sub-percent even at moderately
sub-dilute sizes. The sign structure of the residuals (from above at the
smallest $n$, from below at larger $n$, with a non-monotone crossover at
$n=4$, $d=2$) indicates competing subleading corrections whose form we do not
attempt to fix; what the data establish is the leading term. We record
Eq.~\eqref{eq:sharp_constant} as a conjectured sharp asymptotic. Its two
factors have a mechanistic reading that connects directly to
Section~\ref{subsec:tightness}: the factor $d$ is the unitarity-pinned column
weight---for a column of $X$ labelled by a bunch at site $u_0$, the squared
norm is $d\sum_{v\,\mathrm{admissible}}(n_v+1)\,\lvert U_{vu_0}\rvert^2\to d$
in the dilute limit, fixed by the normalization of the $u_0$-th column of $U$
rather than fluctuating---and the factor $n-d+1$ is the number of excitations
available to complete a $d$-fold bunch from a site holding $d-1$, i.e.\ the
in-degree of the hub structure, entering coherently. The numerics thus answer
the question left open in Section~\ref{subsec:tightness}: the hub Rayleigh
quotient is not destroyed by interference among the $O(m)$ outgoing
amplitudes, the norm is $\Theta(\sqrt n)$ at fixed $d$, and the
polylogarithmic factor in Theorem~\ref{thm:leakage} is an artifact of the
moment method, absent from the actual spectral edge.

Second, the Haar-to-Gaussian comparison of Lemma~\ref{lem:weingarten} is
consistent in sign and trend, though not in magnitude: at the accessible sizes
the moment order $k\sim n\log m$ gives $k^4/m^2\gtrsim1$, so the numerics probe
only the sign of the correction, not the predicted $\tilde O(n^4/m^2)$ scale. At
every size tested the Haar norm lies
\emph{below} the Gaussian one---the direction of the one-sided inequality
Eq.~\eqref{eq:weingarten}---and the ratio increases monotonically toward unity
as $m$ grows at fixed $n$, from $0.81$ at $(n,m)=(3,12)$ to $0.91$ at
$(3,192)$. The mechanism is again the unitary constraint: the column
normalization $\sum_v\lvert U_{vu_0}\rvert^2=1$ removes the upper tail of the
Gaussian column-norm fluctuations, so the Gaussian model over-estimates the
Haar norm, and the mode-scaling law derived from it is, if anything,
conservative.

Third, the Hermitian constraint on the coupling matrix is verified to be
immaterial at the level of the norm: the independent-entry and Hermitian
Gaussian ensembles agree within their interquartile ranges at every grid
point, as do the Haar-entry and Hermitian $i\log U$ ensembles, confirming the
reduction stated in Section~\ref{sec:framework}.

\begin{figure*}[t]
  \centering
  \includegraphics[width=\textwidth]{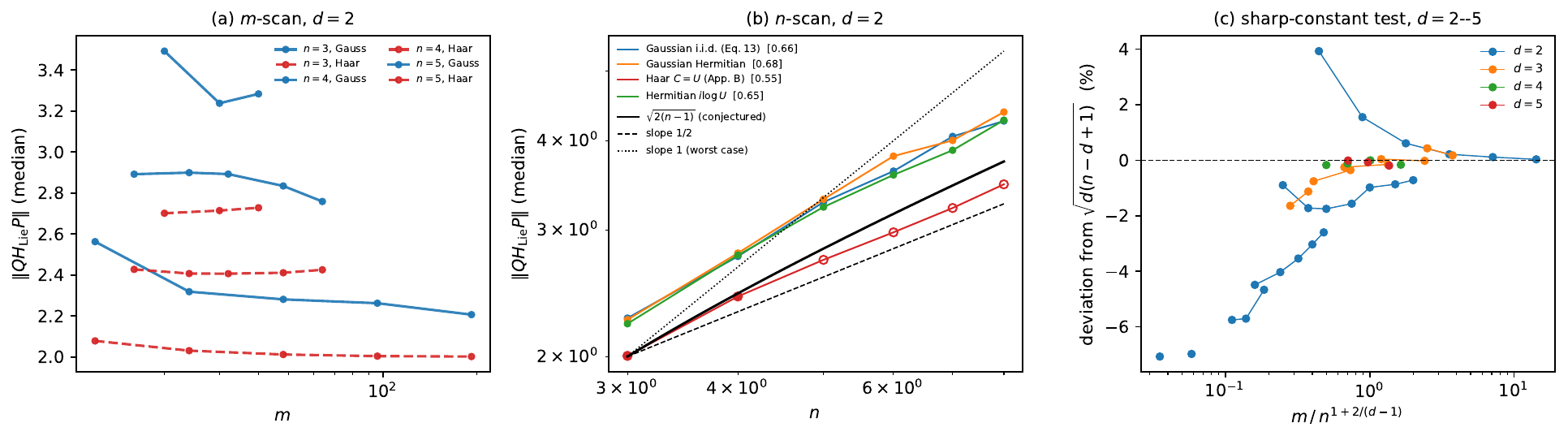}
  \caption{Numerical verification of the leakage bound. \textbf{(a)}~Median
  leakage norm against mode number $m$ at fixed $n$ ($d=2$): the Haar-ensemble
  norm (dashed) is nearly $m$-independent and tightly concentrated, while the
  Gaussian model (solid) drifts toward it from above. \textbf{(b)}~Median norm
  against particle number $n$ ($d=2$, log-log) for all four coupling
  ensembles; guide lines show slope $1/2$ (Theorem~\ref{thm:leakage}) and
  slope $1$ (the worst-case estimate of
  Ref.~\cite{peropadre2017equivalence}); the black curve is the conjectured
  sharp asymptotic $\sqrt{d(n-d+1)}$ of Eq.~\eqref{eq:sharp_constant}. Open
  markers denote sub-dilute Haar points ($m<2n^2$), which fall below the
  dilute-limit curve, in the direction that strengthens the bound.
  \textbf{(c)}~Relative deviation of the Haar median from $\sqrt{d(n-d+1)}$
  for $d=2,3,4,5$, against the dilution parameter $m/n^{1+2/(d-1)}$; each line
  is one $(d,n)$ series in $m$: sub-percent agreement throughout the dilute
  regime, converging with $m$, with visibly tighter agreement at larger $d$.
  Fluctuations over independent realizations are smaller than the markers for
  the Haar ensemble; seeds are deterministic per grid point.}
  \label{fig:numerics}
\end{figure*}

\section{Physical Realization and the Connectivity Bottleneck}
\label{sec:experiment}

The mode-scaling law of Section~\ref{subsec:scaling} rests on a premise intrinsic
to analog Hamiltonian emulation: the evolution time needed to generate the
sampling distribution is a constant, $\tau=\mathcal O(1)$. This is justified
because generalized boson sampling requires only single-particle scrambling, not
full many-body Haar randomness---a dense, sufficiently mixed single-particle
matrix $A$ produces the required interference in a time independent of system
size. On hardware that natively supports a fully connected coupling matrix, this
target is met directly (for instance $\tau=\pi/2$), and the spatial overhead is
then dictated purely by the local dimension $d$, making
$m=\tilde\Omega(n^{1+2/(d-1)})$ a standalone bottleneck.

Translating this constant-time premise to restricted hardware topologies is where
the difficulty lies, because it sets the algorithmic target time against the
physical scrambling time. For a simulator constrained to one- or two-dimensional
local interactions, the finite propagation speed of information---bounded by the
Lieb-Robinson velocity $v_{LR}$~\cite{Richerme2014}---means global scrambling
cannot occur instantaneously. Reproducing the dense algorithmic transitions then
requires a physical time $\tau_{\mathrm{phys}}$ proportional to the diameter of
the interaction graph: $\tau_{\mathrm{phys}}=\Omega(m)$ for a 1D lattice,
$\Omega(\sqrt m)$ for a 2D planar array. On a platform with non-local
connectivity the two coincide, $\tau_{\mathrm{phys}}=\tau=\mathcal O(1)$, whereas
on a local one $\tau_{\mathrm{phys}}$ grows with the diameter and replaces the
constant $\tau$ in the bound. Because the leakage dynamics of the
Dyson series run throughout the physical simulation, this geometry-dependent
duration must be substituted back into the error bound,
$\varepsilon\le\tau_{\mathrm{phys}}\,\lVert QH_{\mathrm{Lie}}P\rVert\,
\max_t\lVert P\lvert\phi(t)\rangle\rVert$. The extended evolution then provides a
window for the state to leak into the bunched sectors, and for finite local
dimension ($d\in\{2,3\}$) the truncation error is no longer suppressed by spatial
dilution---it grows with the macroscopic scrambling duration.

This disqualifies these architectures for \emph{analog, continuous-time}
emulation of the sampling dynamics (a digital compilation into swap networks is a
separate question, not addressed here), such
as one-dimensional spin chains or strictly planar superconducting grids with
local capacitive couplings. A viable architecture must instead possess non-local
connectivity: the interaction graph must support long-range or all-to-all
couplings that decouple the physical scrambling time from the system size. We
evaluate state-of-the-art platforms through this single criterion, the
combination of finite-dimensional fidelity and non-local connectivity.

\subsection{The trade-off: deterministic SPAM versus topological penalty}
\label{subsec:tradeoff}

The criterion is best understood as a trade-off against the photonic baseline. In
optical boson sampling the medium is itinerant photons, whose local Hilbert space
is infinite ($d\to\infty$), so truncation errors are absent. This immunity lets
the architecture reach the dense algorithmic baseline of linear boson sampling
with minimal spatial resources, $m=\Theta(n)$ (e.g.\ $m\ge 2.1n$~\cite{bouland2026linear}), and remain
protected from error growth during the $\mathcal O(m)$-depth interferometric mesh,
since the leakage rate is identically zero. The cost lies elsewhere: probabilistic
generation of large $n$-photon input states, and the exponential sensitivity of
itinerant photons to optical loss in deep interferometers, which together limit
the fidelity and verifiability of the output.

Finite-dimensional matter-based emulation resolves this by moving the
computational degrees of freedom from flying photons to stationary elements,
yielding near-deterministic state preparation and measurement (SPAM):
configuration initialization and number-resolved readout proceed with high
fidelity, bypassing probabilistic generation and mitigating particle loss. The
payment is spatial and topological. A finite local dimension $d$ means that any
$d$-fold bunching event breaks unitarity, so the density must be diluted, raising
the mode count to $m=\tilde\Omega(n^{1+2/(d-1)})$; and because the coupling matrix
is synthesized over time, suppressing leakage demands fast-scrambling interaction
graphs that keep the physical duration short. The pursuit of finite-dimensional
emulation thus trades the problem of synthesizing indistinguishable particles for
that of fabricating highly connected graphs to house diluted excitation networks.
Viewed through this dual constraint, platforms optimized for local operations
(standard planar superconducting grids with nearest-neighbour couplings) are
misaligned with the task, while layouts that natively support long-range or
all-to-all interactions can saturate the optimal scaling.

\subsection{Trapped-ion crystals: all-to-all connectivity}
\label{subsec:ions}

Trapped atomic ions in radio-frequency Paul traps or static Penning traps are an
exceptionally clean platform for finite-dimensional sampling~\cite{blatt2012quantum}. Isolating the
internal hyperfine or optical-clock states of isotopes such as
$^{171}\text{Yb}^+$ or $^{40}\text{Ca}^+$ gives stationary qudit registers with
SPAM fidelities routinely exceeding $99.9\%$~\cite{harty2014highfidelity,wang2021single}, bypassing the probabilistic
generation and transmission losses of photonic networks. The multi-level Zeeman
or hyperfine manifolds natively accommodate higher local dimensions: the
$|F=1,m_F=-1,0,1\rangle$ ground manifold of $^{171}\text{Yb}^+$ encodes a spin-1
($d=3$) representation via multi-tone microwave or Raman driving. The $m_F=\pm1$
states are susceptible to linear Zeeman shifts, but this can be protected against
magnetic dephasing by continuous dynamical decoupling or multi-ion entangled
encodings.

The decisive advantage of this platform is that it bypasses the Lieb-Robinson
constraint through phonon-mediated long-range interactions. Spin-dependent optical
dipole forces detuned from the collective transverse motional modes synthesize an
effective long-range $XY$ hopping Hamiltonian~\cite{porras2004effective,monroe2021programmable}
$H_{\mathrm{Lie}}=\sum_{i<j}J_{ij}(E^{(j)}_+E^{(i)}_-+\mathrm{H.c.})$, with
\begin{equation}
    J_{ij} = \sum_p \frac{\Omega_i\Omega_j\,\eta^p_i\eta^p_j\,\omega_p}
                       {\mu^2-\omega_p^2}
    \approx \frac{J_0}{|R_i-R_j|^\alpha},
    \label{eq:ion_coupling}
\end{equation}
where $\Omega_i$ is the single-ion Rabi frequency, $\eta^p_i$ the Lamb--Dicke
parameter coupling ion $i$ to the $p$-th motional mode $\omega_p$, and $\mu$ the
Raman detuning. The exponent $\alpha$ is continuously tunable~\cite{britton2012engineered} between $0$ and $3$
through $\mu$: operating near the centre-of-mass mode sends $\alpha\to0$,
collapsing the coupling to a fully connected, all-to-all graph with diameter
$\mathrm{diam}(G)=1$. The physical scrambling time is then decoupled from the mode
count and matches the algorithmic target $\tau=\mathcal O(1)$. Isolating the
centre-of-mass mode without off-resonantly exciting higher spatial modes requires
keeping the base coupling weak, $J_0/2\pi\sim 1$--$5\,$kHz; this slightly extends
the absolute simulation window, but it is completed well within the coherence
times of dynamically decoupled ions, $T_2\ge 100\,$ms~\cite{wang2021single}. With this
geometry-invariant duration, the integrated truncation error is suppressed by the
spatial-volume expansion and the scaling law is fully operative.

Linear Paul traps currently sustain co-linear chains of $m\sim 50$--$100$ ions
(e.g.\ the JQI and Innsbruck
architectures~\cite{zhang2017observation,joshi2022observing}) with individual addressing and
single-site readout, and two-dimensional Paul-trap crystals have recently
reached more than $500$ ions with site-resolved readout~\cite{guo2024site};
two-dimensional Penning traps (the NIST
framework~\cite{britton2012engineered,bohnet2016quantum}) hold up
to $m\sim 200$ ions in a planar crystal. Maintaining stable individual addressing
across large 2D arrays carries alignment and phase-stability overheads, but the
suppression of truncation leakage outweighs them, making trapped ions a leading
candidate.

\subsection{Superconducting circuits: Kerr truncation and non-local routing}
\label{subsec:cqed}

Circuit QED offers a programmable solid-state substrate. Dispersive interactions
between Josephson-junction artificial atoms and microwave resonators give SPAM
fidelities exceeding $99\%$, bypassing probabilistic input generation. Using bare
transmons to simulate generalized bosons, however, forces the most penalized
scaling: the weak transmon anharmonicity ($\alpha/2\pi\approx-300\,$MHz)~\cite{koch2007charge} confines
it to a hard-core two-level system ($d=2$), invoking the maximal overhead
$m=\tilde\Omega(n^3)$.

To reach the efficient $m=\tilde\Omega(n^2)$ boundary of higher-spin
representations ($d\ge 3$), the degrees of freedom must move from the transmons to
the continuous-variable cavities. A bare cavity is harmonic ($d\to\infty$), which
would revert the emulation to canonical boson sampling and nullify the
Lie-algebraic mapping; to enforce a finite-$d$ truncation, dispersively coupled
transmons engineer a strong Kerr nonlinearity ($K/\kappa\gg1$)
\cite{leghtas2015confining,mirrahimi2014dynamically}. This Kerr blockade prohibits
the cavity from occupying $|n\ge d\rangle$ Fock states. The transmons then play a
dual role---executing photon-number-resolved selective $\pi$-pulses to initialize
generalized Fock states~\cite{touzard2019gated,heeres2015cavity}, and enforcing the truncation---
while the emulation proceeds in the low-order cavity manifolds, which have
intrinsic lifetimes $T_1\ge 1\,$ms~\cite{reagor2016quantum,milul2023superconducting}.

Viability is then governed by topology. Prevailing architectures use planar grids
optimized for nearest-neighbour surface-code correction~\cite{krinner2022realizing}, with capacitive couplings
bounded at $J/2\pi\sim 10$--$20\,$MHz. For analog sampling this 2D geometry incurs
a Lieb-Robinson limit: global scrambling needs
$\tau_{\mathrm{phys}}=\Omega(\sqrt m/J)$, whose substitution into the error bound
drives the truncation leakage up with system size. Standard planar lattices are
therefore disqualified from scalable continuous-time sampling. The remedy is to
abandon purely local planar geometries for non-local routing, which is viable in
3D cQED: coupling multiple high-$Q$ storage cavities to a central multimode bus---
such as the 3D ``flute'' cavities of recent quantum-memory
experiments~\cite{chakram2021seamless}---and
applying parametric multipump wave-mixing~\cite{naik2017random} synthesizes an effective all-to-all
hopping $J_{ij}^{\mathrm{eff}}$ by modulating tunable couplers at the mode-pair
difference frequencies. Provided the bus sustains the multiplexed drives without
severe Purcell decay or spurious multiphoton resonances, this restores
$\tau_{\mathrm{phys}}=\mathcal O(1)$, making non-local 3D cQED the viable
solid-state route.

\subsection{Neutral-atom arrays and cavity-mediated networks}
\label{subsec:neutral}

Neutral atoms in optical tweezer arrays are highly scalable and resolve the SPAM
bottlenecks of linear optics, with single-site preparation and fluorescence-
imaging fidelities exceeding $99\%$~\cite{barredo2016atom,endres2016atom,manetsch2025tweezer}. The hyperfine manifolds of alkaline-earth-
like atoms ($^{88}\text{Sr}$ or $^{171}\text{Yb}$)~\cite{cooper2018alkaline,norcia2018microscopic,saskin2019narrow} provide a decoherence-free
encoding for higher-dimensional representations: addressing long-lived magnetic
sublevels by multi-tone Raman transitions isolates $d=3$ (spin-1) or $d=4$
(spin-3/2) subalgebras, allowing deterministic multi-excitation state preparation
without the dispersive distortions that limit superconducting resonators.

The vulnerability is again topological. Standard coupling via the Rydberg blockade
relies on van der Waals interactions decaying as $1/R^6$~\cite{saffman2010quantum,browaeys2020many}, which restrict the
interaction graph to local geometries. By the connectivity criterion, executing a
globally scrambled $A$ on such a graph forces the evolution time to grow with the
array diameter, opening the leakage window before the sampling regime is reached.
Two-dimensional van der Waals arrays therefore succumb to error growth. A partial
remedy uses resonant dipole-dipole interactions between Rydberg states, with a
longer-range $1/R^3$ tail~\cite{barredo2015coherent}; reaching the fast-scrambling $\tau=\mathcal O(1)$
target, however, ultimately requires embedding the tweezer array in a
high-finesse optical cavity. Cavity QED supplies a photon-mediated bus coupling
all atomic ensembles irrespective of separation: driving the global cavity mode
with detuned optical fields synthesizes an effective all-to-all hopping across the
mapped internal states~\cite{periwal2021programmable}, removing the planar Lieb-Robinson bound and granting the
required constant scrambling time. Cavity-coupled neutral-atom arrays thus combine
deterministic higher-spin SPAM with non-local connectivity.

\subsection{Local-interaction platforms: optical lattices and integrated
photonics}
\label{subsec:local}

Two further platforms fall on the unfavourable side of the criterion, for opposite
reasons, and together delimit it.

Ultracold atoms in optical lattices are the canonical analog
simulator~\cite{bloch2012quantum}, natively
realizing Bose--Hubbard dynamics and supporting very large mode numbers,
$m\sim10^4$ spatial sites, with $J$ and $U$ tunable by Feshbach resonances and
lattice depth~\cite{chin2010feshbach,gross2017quantum} ($J/h\sim 100\,$Hz--$1\,$kHz). Their utility bifurcates by $d$. In
the weakly interacting limit ($U\ll J$, e.g.\ dilute $^{87}\text{Rb}$) the local
occupation is unconstrained ($d\to\infty$), truncation leakage is absent, and the
lattice is an excellent canonical linear-boson-sampling emulator at $m=\Theta(n)$,
tolerating the macroscopic scrambling time of its local geometry. But for
finite-$d$ emulation---spinless fermions for $d=2$, or the Mott regime
($U\gg J$)~\cite{greiner2002quantum,childs2014bose} for a spin-1 subalgebra ($d=3$)---the
nearest-neighbour topology imposes a rigid Lieb-Robinson cone, $v_{LR}\propto
Ja/\hbar$, so global scrambling of a 2D lattice needs
$\tau_{\mathrm{phys}}=\Omega(\sqrt m/J)$, opening the leakage window into the
$n_k\ge d$ sectors. Standard lattices are thus excellent for $d\to\infty$
continuous variables but geometrically preclude finite-$d$ emulation, unless
augmented by a non-local cavity bus, which returns to the all-to-all requirement
above.

Integrated photonic circuits sit at the opposite pole. Programmable Mach--Zehnder
meshes on silicon nitride or lithium niobate synthesize arbitrary unitaries over
$m\sim100$ modes via Clements or Reck layouts~\cite{reck1994experimental,clements2016optimal,carolan2015universal}, and with infinite local dimension
($d\to\infty$) the truncation projector $P$ vanishes identically, so the platform
trivially bypasses the Lieb-Robinson penalty and targets $m=\Theta(n)$. Here the
binding constraint is not topological but the exponential loss barrier: executing a
Haar-random unitary sends the photons through an $\mathcal O(m)$-depth mesh, and
even at component transmissivity $\eta_0>99\%$ the $n$-photon survival probability
decays as $\eta_0^{\mathcal O(nm)}$. When a macroscopic fraction of photons is
lost, the output distribution transitions from a $\#$P-hard permanent- or
hafnian-based distribution to an efficiently simulable classical one; Gaussian
boson sampling with deterministic squeezed sources mitigates the generation
bottleneck but not the propagation loss. Integrated photonics thus realizes the
ideal zero-truncation emulation, bottlenecked not by scrambling but by the
exponential suppression of the multi-particle coincidence rate.

\section{Conclusion and Outlook}
\label{sec:conclusion}

We have established a quantitative resource law for emulating boson sampling on
deterministic, finite-dimensional hardware. Placing boson, fermion, and spin
samplers in a single framework---non-interacting evolution on an irreducible
representation, with the transition amplitude given by a submatrix immanant---we
analyzed the bunching leakage that the local truncation introduces and that is
absent in linear optics. The central technical result, Theorem~\ref{thm:leakage},
is that the leakage operator norm concentrates at $\tilde O(\sqrt n)$ rather than
at the worst-case $O(n)$ of prior spin-based emulations, proved by decomposing the
correlated many-body operator into independent random matrices and applying
non-commutative concentration in a Gaussian model of the transition matrix; the
transfer to the physical Haar ensemble is reduced to a comparison step
(Conjecture~\ref{conj:transfer}, Appendix~\ref{app:weingarten}) that we verify at
leading order but leave open in general. This tightens the mode
requirement from $m=\Omega(n^4)$ to $m=\tilde\Omega(n^{1+2/(d-1)})$, which for a
spin-$1$ local representation ($d=3$) falls to $m=\tilde\Omega(n^2)$, coinciding
with the collision-free threshold of the dilute limit itself. We do not establish
new hardness; what we provide is the spatial cost---controlled by the local
dimension $d$, and operative only on platforms with fast, non-local
connectivity---of preserving the hardness inherited from boson sampling.

Two questions are left open, and the body of this discussion is given to them
because they are not independent. The first is the removal of the polylogarithmic
factor in the upper bound. The bound of Theorem~\ref{thm:leakage} is
$\tilde O(\sqrt n)$, optimal up to a polylogarithmic factor in the mode count, and
the origin of that factor is precise enough to point at what a sharper result
would need. We access the operator norm through the high moments
$\mathbb E\,\mathrm{Tr}\big((XX^\dagger)^k\big)$ and pass to the norm by Markov's
inequality at a fixed order $k$; the prefactor in the concentration inequality is
the entropic dimension of the Fock space, $\log D=O(n\log m)$, and to make the
exponent dominate it the order must be taken as $k\sim n\log m$. The logarithms in
the final bound are the price of reaching the spectral edge through the $k$-th
moment rather than controlling it directly, so a method that bounded the largest
singular value without passing through a high moment would remove them.

Three routes toward such a method suggest themselves, none closed for the
structured operator at hand. The most direct is an edge-universality statement:
for unstructured ensembles the extreme singular value is pinned to the
Marchenko--Pastur edge with sharp Tracy--Widom fluctuations and no logarithmic
correction, but edge universality is established under independence or weak
dependence of the entries, whereas the leakage operator has entries fixed by the
$O(m^2)$ parameters of $A$ and constrained by the projectors $P,Q$, making it both
correlated and sparse; extending edge universality to operators of this
projector-constrained type is, to our knowledge, open. A second route is a local
law---the resolvent methods controlling $(XX^\dagger-z)^{-1}$ to the optimal scale
near the edge, which have been carried to sparse and structured ensembles such as
random regular graphs and Erd\H{o}s--R\'enyi matrices in other contexts; adapting
one to the bipartite support graph of the leakage operator would give edge control
without the moment detour, the difficulty being that this graph is not a standard
sparse random graph but a strongly asymmetric one. The third route is spectral
rather than entry-level, and is suggested by a neighbouring result: for the traces
of the first powers of a Haar unitary, $\mathrm{Tr}\,U^k$, the rate of Gaussian
approximation in total variation is super-exponential, with explicit dependence on
the number of traces and the dimension~\cite{johanssonlambert2020}, showing that
for spectral objects a precision far beyond the moment method is attainable when
the relevant degree grows with the dimension. The leakage operator is an
entry-level object, a polynomial in the matrix elements rather than a function of
the eigenvalues, and the two are governed by different machinery---the entry-level
transfer is the Weingarten regime, with polynomial total-variation rates, while
the spectral transfer admits the super-exponential rate---so whether an
entry-level analogue exists for structured operators is open; if it did, it would
supply exactly the edge-level control the moment method lacks.

The second open question is the matching lower bound, and it turns out to be the
same question seen from the other side. Theorem~\ref{thm:leakage} bounds the norm
from above; whether $\lVert QH_{\mathrm{Lie}}P\rVert=\tilde\Theta(\sqrt n)$, with a
matching $\Omega(\sqrt n)$, decides whether the mode-scaling law is optimal or
merely sufficient. It is unresolved analytically; numerically, the
Haar-ensemble norm matches the closed form $\sqrt{d(n-d+1)}$ of
Section~\ref{subsec:numerics} to sub-percent accuracy across $d=2$--$5$,
fixing the expected answer together with its sharp constant. The concentration inequality is one-sided
by construction: it bounds the norm from above through the variance statistic, and
the same dilution that makes the variance small offers no purchase on the extreme
singular value from below, so a lower bound must come from the geometry of the
operator rather than from a global average. That geometry is the strongly asymmetric
bipartite support graph of Section~\ref{subsec:tightness}, and the two routes to a
matching bound start from its high-degree hubs exactly as described there: a
hub-localized Rayleigh quotient, which must overcome the destructive interference
among the $O(m)$ outgoing amplitudes, and a refined, hub-dominated moment method,
which inherits the same edge-versus-bulk limitation as the upper bound. We do not
repeat that discussion here.

That shared limitation is what makes the two questions one. The polylogarithm is
present because we reach the norm through the $k$-th moment rather than through the
spectral edge, and the lower bound is missing because the moment method sees the
bulk and not the edge; a direct spectral analysis of the structured leakage
operator---an edge-universality or local-law statement adapted to its sparse,
correlated, asymmetric support graph---would remove the polylog and supply the
matching lower bound at once, replacing $\tilde O(\sqrt n)$ and the open
$\Omega(\sqrt n)$ by a single $\Theta(\sqrt n)$. Both reduce to controlling the
spectral edge of Fock-space operators built from a small number of correlated
parameters, which we isolate as the principal mathematical question left by this
work; until it is resolved, the result stands as a near-optimal upper bound,
sufficient to place finite-dimensional emulation within the hardness regime.

Turning from the mathematics to the laboratory, the connectivity criterion of
Section~\ref{sec:experiment} sharpens what a demonstration would require, and the
obstacles are concrete. The first is the overhead itself: for hard-core qubits
($d=2$) the requirement $m=\tilde\Omega(n^3)$ is severe, since even $n\sim 50$
would demand of order $10^5$ modes, beyond any near-term array, and this is the
operational reason the spin-$1$ result matters---at $d=3$ the requirement drops to
$m=\tilde\Omega(n^2)$, bringing the same $n\sim 50$ to a few thousand modes, large
but within the trajectory of trapped-ion and neutral-atom platforms. The practical
path therefore runs through higher local dimension, which makes protecting the
$d=3$ encoding (the $m_F=\pm1$ sublevels of $^{171}\text{Yb}^+$, or the magnetic
sublevels of alkaline-earth atoms) against dephasing a central rather than
incidental requirement. The second obstacle is maintaining non-local connectivity
as $m$ grows: the weak base coupling needed to isolate a single collective mode
($J_0/2\pi\sim$ a few kHz for ions), the addressing stability across a 2D crystal,
and the susceptibility of a shared microwave bus to Purcell decay all degrade with
scale, so that while the criterion is binary in principle, in practice the
all-to-all coupling degrades continuously and the operative question is the
largest $m$ at which $\tau$ stays effectively constant.

The third obstacle, verification, is the most fundamental, because it inherits the
very hardness the protocol relies on. Certifying the output against the ideal
distribution requires evaluating that distribution, which is the permanent and
hence \#P-hard, so cross-entropy benchmarking---which scores samples by their
average ideal probability $F_{\mathrm{XEB}}=\langle 2^n P_{\mathrm{ideal}}(x)
\rangle-1$---demands a classical computation feasible only at small $n$; the same
ceiling limits the likelihood-ratio and Bayesian tests developed for photonic
boson sampling~\cite{aaronson2014uniform,spagnolo2014experimental,carolan2014verification}, while the heavy-output and suppression-law checks~\cite{tichy2014stringent} probe necessary
but not sufficient features of the distribution, trading feasibility for
certification strength. Scalable certification of boson sampling,
finite-dimensional or photonic, remains open~\cite{hangleiter2023computational,hangleiter2019sample}, and what is available is validation
at moderate $n$, where the permanent is still computable, together with
consistency checks that scale. Within that validatable regime, finite-dimensional
emulation faces a specific threat to benchmarking and a specific remedy, both made
precise by the framework. The threat is the residual superexchange: the truncation
is exact only as $U\to\infty$, and at finite $U$ virtual transitions to the gapped
sector generate $O(J^2/U)$ density-density or $Z_iZ_j$ couplings. These are a
coherent, $U(1)$-preserving perturbation of the ideal hopping generator
$H_{\mathrm{Lie}}$, not a source of new hardness---the \#P-hardness rests on the
permanent structure of the amplitude, and the effect of the superexchange is
bounded directly by the Dyson analysis of Section~\ref{sec:mainresult}, with no
appeal to any spectral or ergodic property---but for verification an
\emph{uncharacterized} static perturbation of this kind is fatal, since it biases
the ideal probabilities used in the benchmark and corrupts the fidelity estimate.
The remedy is that the perturbation is learnable. Late-time tensor-network
tomography is unviable on the fast-scrambling graphs the protocol needs, the
volume-law entanglement defeating efficient classical representation; but the
superexchange is static and few-body, so a short-time Hamiltonian-learning
protocol---evolving simple product states for a duration
$\delta t\ll\tau_{\mathrm{scramble}}$ and measuring localized two-body
correlators---isolates the $O(J^2/U)$ matrix elements before scrambling sets in,
and incorporating them into the classical benchmark model separates the emulation
fidelity from the calibration error, making the benchmark usable at moderate $n$
despite the residual interactions. Beyond that small-$n$ window the framework
supplies a check that scales: the sampling probability of Eq.~\eqref{eq:prob}
factorizes into the ideal boson-sampling probability and the bunching penalty
$\chi_{\mathfrak g}$, concentrating the entire finite-dimensional deviation in a
single measurable quantity, the rate of $d$-fold collisions, so that confirming
this rate sits below the threshold predicted by the mode-scaling law---weaker than
full certification, but accessible at any $n$---verifies that the emulator
operates in the regime where the hardness argument applies. The same factorization
clarifies how leakage enters the benchmark, a drop in score being attributable to
the collision penalty $\chi_{\mathfrak g}$ rather than to generic depolarizing
noise, which is a sharper diagnostic than the score alone; whether this makes
finite-dimensional sampling more or less susceptible to the classical spoofing
strategies recently raised against cross-entropy benchmarking, which exploit the
noise-induced decay of the score to forge it efficiently, is an open question, and
a natural one given that the collision penalty is a coherent rather than a
stochastic deviation.

Several extensions follow naturally. The analysis was scoped to the unitary family
in its defining representation, where the amplitude is a permanent; for the
orthogonal and symplectic families the relevant commutant is the Brauer algebra
and the amplitudes are Brauer-algebra invariants rather than $S_N$-immanants
(Section~\ref{sec:framework}), so determining when an $SO(M)$ or $Sp(M)$ defining
representation yields a genuinely hard invariant, and whether the leakage analysis
transfers to the Brauer setting, would extend the resource law beyond the
permanent case. The lower bound discussed above is a second direction, its
resolution converting the near-optimal statement into an optimal one. A third is
the connection to lossy and noisy regimes: the bunching leakage analyzed here is a
coherent, unitarity-preserving error internal to the truncation, distinct from the
incoherent particle loss that limits photonic platforms, and a unified treatment
of the two error channels on a single matter-based architecture remains to be
developed.

\section{Acknowledgements}
EJK acknowledges financial support from the National Science and Technology Council (NSTC) of Taiwan under Grant No.~NSTC~114-2112-M-A49-036-MY3

\emph{Data and code availability.}---The numerical code (a self-contained,
seed-deterministic notebook), the raw norm data underlying
Fig.~\ref{fig:numerics}, and the figure-generation scripts are available at
[Zenodo DOI to be inserted]; the complete parameter grid and sector dimensions
are tabulated in Table~\ref{tab:grid}.

\appendix

\section{Algebraic Factorization and Schur--Weyl Projection on Tensor Spaces}
\label{app:tensor_factorization}

In Sec.~\ref{sec:framework} we used the fact that a non-interacting many-body
evolution operator acts as the $N$-fold tensor power of the defining-representation
matrix, restricted to a symmetry subspace. This appendix derives that statement in
full and identifies precisely where it specializes to the unitary ($A$-series) family.

\subsection{Tensor action of the algebra}
Let $\mathfrak{g}$ be a compact Lie algebra and $V\cong\mathbb{C}^M$ its defining module,
on which $H\in\mathfrak{g}$ acts as $\pi_{\mathrm{def}}(H)$. The unsymmetrized
$N$-particle space is $V^{\otimes N}$. On a tensor product the algebra acts through the
primitive coproduct, i.e.\ as a derivation distributing over the factors,
\begin{equation}
    \pi_{\otimes N}(H) = \sum_{k=1}^N H^{(k)},
    \qquad
    H^{(k)} = \mathbb{I}\otimes\cdots\otimes
    \underbrace{\pi_{\mathrm{def}}(H)}_{k\text{-th slot}}\otimes\cdots\otimes\mathbb{I}.
    \label{eq:app_symmetric_sum}
\end{equation}
Since $H^{(k)}$ and $H^{(l)}$ act on disjoint slots, $[H^{(k)},H^{(l)}]=0$ for $k\neq l$.
This additive (derivation) action at the level of the algebra exponentiates to a
\emph{multiplicative} (group-like) action at the level of the group, which is what turns
the sum~\eqref{eq:app_symmetric_sum} into a tensor \emph{power} below.

\subsection{Exponentiation to a tensor power}
Because the summands in Eq.~\eqref{eq:app_symmetric_sum} commute, the exponential
factorizes,
\begin{equation}
    \hat{\mathcal{U}}_{\mathrm{full}}(t)
    = \exp\!\big[-i\,\pi_{\otimes N}(H)\,t\big]
    = \prod_{k=1}^N \exp\!\big[-i\,H^{(k)}t\big].
    \label{eq:app_exponential_factorization}
\end{equation}
Each factor is the identity on every slot but the $k$-th, where it equals the
single-particle propagator $U=\exp[-i\,\pi_{\mathrm{def}}(H)\,t]$, an $M\times M$ matrix.
The product is therefore the tensor power
\begin{equation}
    \hat{\mathcal{U}}_{\mathrm{full}}(t) = U\otimes U\otimes\cdots\otimes U
    \equiv U^{\otimes N}.
    \label{eq:app_tensor_power}
\end{equation}
On the full unsymmetrized space the many-body evolution is thus exactly $U^{\otimes N}$
with $U=\pi_{\mathrm{def}}(e^{-iHt})$. This holds for any $\mathfrak{g}$ and any defining
module; the family-dependent input enters only through the symmetry projection of the
next subsection.

\subsection{Schur--Weyl projection and the amplitude (unitary family)}
Physical states occupy a symmetry subspace fixed by the exchange statistics. For the
unitary family $G=U(M)$, the commutant of $U^{\otimes N}$ on $V^{\otimes N}$ is the group
algebra $\mathbb{C}[S_N]$ acting by slot permutations $\hat P_\sigma$, and
\begin{equation}
    [U^{\otimes N},\,\hat P_\sigma]=0,\qquad \forall\,\sigma\in S_N,
    \label{eq:app_schur_weyl_commute}
\end{equation}
so $V^{\otimes N}=\bigoplus_{\lambda\vdash N,\,\ell(\lambda)\le M}
\mathbb{S}^{\lambda}\otimes W_\lambda$, with $\mathbb{S}^\lambda$ the Specht module
($\dim d_\lambda$) and $W_\lambda$ the Weyl module.

A point that must be handled with care for mixed diagrams: the central projector
$\Pi_\lambda=\tfrac{d_\lambda}{N!}\sum_\sigma\chi^\lambda(\sigma)\hat P_\sigma$ projects
onto the entire isotypic block $\mathbb{S}^\lambda\otimes W_\lambda$, of dimension
$d_\lambda\dim W_\lambda$, and so selects a single state only when $d_\lambda=1$. A
physical weight vector for general $\lambda$ requires, in addition to the occupation
data, a choice of standard Young tableau $T$, i.e.\ a primitive Young idempotent
$\hat e_T$ ($\hat e_T^2=\hat e_T$, $\Pi_\lambda=\sum_T \hat e_T$) that picks one copy of
$W_\lambda$ from the multiplicity space:
\begin{equation}
    |\mathbf{m};T\rangle = \mathcal{N}_{\mathbf{m}}\,\hat e_T\,|I_{\mathbf{m}}\rangle,
    \qquad
    |I_{\mathbf{m}}\rangle=\bigotimes_{k=1}^N|i_k\rangle,
    \label{eq:app_young_idempotent}
\end{equation}
where $\mathcal{N}_{\mathbf{m}}$ is the occupation-dependent normalization (carrying the
$\prod_i n_i!$ factors restored in any probability). At the symmetric
($\lambda=(N)$) and antisymmetric ($\lambda=(1^N)$) endpoints $d_\lambda=1$, the tableau
label is vacuous, and $\hat e_T=\Pi_\lambda$.

Using $[U^{\otimes N},\hat P_\sigma]=0$ (hence $[U^{\otimes N},\hat e_T]=0$), and that
$U^{\otimes N}$ acts as the identity on the $\mathbb{S}^\lambda$ factor (Schur's lemma)
so the result is independent of $T$, the weight-to-weight amplitude is
\begin{align}
    \mathcal{A}_{\mathbf{m}\to\mathbf{n}}
    &= \langle I_{\mathbf{n}}|\,U^{\otimes N}\,\hat e_T\,|I_{\mathbf{m}}\rangle
       \,\big/\,\mathcal{Z}_{\mathbf{m},\mathbf{n}} \nonumber\\
    &\propto \sum_{\sigma\in S_N}\chi^{\lambda}(\sigma)
       \prod_{k=1}^N U_{j_k,\,i_{\sigma(k)}}
     = \mathrm{Imm}^{\lambda}\!\big(U[\mathbf{m},\mathbf{n}]\big),
    \label{eq:app_amplitude_final}
\end{align}
where $\mathcal{Z}_{\mathbf{m},\mathbf{n}}$ collects the normalizations and
$U[\mathbf{m},\mathbf{n}]$ is the $N\times N$ submatrix selected by the occupation
patterns. The single permutation sum on the right is exact because, after Schur's lemma
removes the multiplicity space, the character $\chi^\lambda$ is the only surviving
$\mathbb{S}^\lambda$ data; equivalently, $\mathrm{tr}_{\mathbb{S}^\lambda}\hat e_T$
returns $\chi^\lambda$. Equation~\eqref{eq:app_amplitude_final} is the matrix immanant.
The single-row diagram gives $\chi^\lambda\equiv 1$ (permanent); the single-column
diagram gives $\chi^\lambda=\mathrm{sgn}$ (determinant).

\subsection{Scope: other classical families}
The reduction~\eqref{eq:app_schur_weyl_commute}--\eqref{eq:app_amplitude_final} is
specific to $G=U(M)$, where the commutant of $U^{\otimes N}$ is $\mathbb{C}[S_N]$. For
the orthogonal and symplectic families the relevant module is not the full $V^{\otimes
N}$ but its trace-free part, and the commutant is the Brauer algebra $B_N(\pm M)$, whose
basis includes contraction (pairing) diagrams in addition to permutations. The resulting
amplitudes are Brauer-algebra invariants of $U[\mathbf{m},\mathbf{n}]$---not
$S_N$-immanants---and the shape/complexity correspondence must be re-derived in that
setting (Sec.~\ref{sec:framework}). Accordingly, the identity ``non-interacting
amplitude $=$ immanant of an $N\times N$ submatrix'' is established here for the unitary
family; it is the canonical case from which $SO$ and $Sp$ depart through their distinct
invariant theory.

\section{Haar-to-Gaussian Comparison via the Weingarten Expansion}
\label{app:weingarten}

This appendix develops the Haar-to-Gaussian comparison of
Lemma~\ref{lem:weingarten} and isolates the one unproven step on which it rests.
The claim under examination is that the leakage operator norm under the exact Haar
ensemble agrees, in the dilute regime, with the one under the independent Gaussian
model, up to a relative correction $\tilde O(n^4/m^2)$ that vanishes for
$m=\tilde\Omega(n^3)$. The argument has four parts. We first reduce
the comparison of operator norms to a comparison of moments at a single optimized
order $k$ (Sec.~\ref{appB:reduction}). We then write both moments through their
permutation expansions---Weingarten for Haar, Wick for Gaussian---and isolate
their difference in an off-diagonal remainder (Sec.~\ref{appB:expansions}). The
core estimate would bound that remainder by organizing it according to the genus of the
underlying gluing (Sec.~\ref{appB:genus}); the one structural input this requires,
that the projectors and local operators do not disturb the genus ordering, is the
single step we do \emph{not} establish in general. We isolate it as
Conjecture~\ref{conj:transfer} and verify it only at leading order, discussing
below why the projector constraints place it in tension with the uniform power
counting it assumes. We close by tracking the (conditional) result back to the
operator norm, uniformly in $k$ (Sec.~\ref{appB:conclusion}).

\subsection{Reduction to a moment comparison}
\label{appB:reduction}

The leakage operator is linear in the single-particle matrix,
\begin{equation}
    X = \sqrt m\,QH_{\mathrm{Lie}}P = \sum_{u=1}^{m}\sum_{v\ne u} A_{vu}\,M_{vu},
    \qquad M_{vu}=Q\,E^{(v)}_{+}E^{(u)}_{-}\,P,
    \label{eq:appB_operator}
\end{equation}
with $A_{vu}=\sqrt m\,U_{vu}$ for $U$ Haar on $U(m)$ (the Haar model) and
$A_{vu}=Z_{vu}$, $Z_{vu}\sim\mathcal{CN}(0,1)$ independent (the Gaussian model).
For any $k\in\mathbb N$ the operator norm is dominated by the $k$-th moment of the
squared singular values,
\begin{equation}
    \lVert X\rVert^{2k}
    = \Big(\max_i s_i^2\Big)^{k}
    \le \sum_i s_i^{2k}
    = \mathrm{Tr}\big[(XX^\dagger)^k\big],
    \label{eq:appB_norm_moment}
\end{equation}
so Markov's inequality gives, for every $k$ and $t>0$,
\begin{equation}
    \Pr\big[\lVert X\rVert\ge t\big]
    \le t^{-2k}\,\mathbb E\,\mathrm{Tr}\big[(XX^\dagger)^k\big].
    \label{eq:appB_markov}
\end{equation}
The bound \eqref{eq:appB_markov} is not optimized by $k\to\infty$; rather, with
$\mathbb E\,\mathrm{Tr}[(XX^\dagger)^k]\le D\,\sigma^{2k}$ for a matrix of
dimension $D$ and variance scale $\sigma^2$, the right-hand side is
$D(\sigma/t)^{2k}$, minimized over integer $k$ at the largest $k$ for which the
geometric factor still decreases, i.e.\ once $t>\sigma$, by taking $k$ as large as
the expansion remains controlled. Setting $t^2=2(1+c')\sigma^2\log D$ makes the
bound $D^{-c'}=\exp(-c'\log D)$ with the optimizing order
\begin{equation}
    k \;\sim\; \log D \;=\; O(n\log m),
    \label{eq:appB_kopt}
\end{equation}
using $\log D=O(n\log m)$ from Eq.~\eqref{eq:logD}. This fixes the moment order at
$k\sim n\log m$; it is the source of the polylogarithmic factor and is held fixed
throughout. Because $X$ is linear in $A$, the expectation in
\eqref{eq:appB_markov} is a balanced Haar (or Gaussian) integral of a monomial of
degree
\begin{equation}
    2q = 2k, \qquad\text{i.e.}\quad q=k\sim n\log m,
    \label{eq:appB_degree}
\end{equation}
the degree being set by the moment order, not by the particle number $n$. A naive
treatment with $q\le n$ would misjudge the Weingarten counting and produce the
wrong exponent; the correct degree is $q=k$.

\subsection{Weingarten and Wick expansions}
\label{appB:expansions}

For index sequences $\mathbf i,\mathbf j$ of length $q=k$, the exact Haar integral
is the Weingarten formula~\cite{collins2006integration,collins2009orthogonal},
\begin{equation}
    \int_{U(m)}\!\prod_{a=1}^k U_{i_a j_a}\prod_{a=1}^k U^*_{i'_a j'_a}\,d\mu(U)
    = \sum_{\sigma,\tau\in S_k}
      \delta_{\mathbf i,\mathbf i'\sigma}\,\delta_{\mathbf j,\mathbf j'\tau}\,
      \mathrm{Wg}(\sigma^{-1}\tau,m),
    \label{eq:appB_wg}
\end{equation}
where $\delta_{\mathbf i,\mathbf i'\sigma}=\prod_a\delta_{i_a,i'_{\sigma(a)}}$.
The Weingarten function admits the exact convergent
series~\cite{collins2006integration}, valid for $m\ge k$,
\begin{equation}
    \mathrm{Wg}(\nu,m)
    = m^{-k}\sum_{\ell\ge 0}(-1)^{\ell} m^{-\ell}\,
      \#\{(\rho_1,\dots,\rho_\ell)\in (S_k^*)^{\ell}:\rho_1\cdots\rho_\ell=\nu\},
    \label{eq:appB_wg_exact}
\end{equation}
with $S_k^*$ the non-identity permutations, whose leading term reproduces the
asymptotic
\begin{equation}
    \mathrm{Wg}(\nu,m) = m^{-(k+|\nu|)}\big(\mathrm{M\ddot ob}(\nu)+O(m^{-2})\big),
    \qquad |\nu|=k-C(\nu),
    \label{eq:appB_wg_asymp}
\end{equation}
where $C(\nu)$ is the number of cycles and $|\nu|$ the minimal transposition
length. The only facts we use are the power counting $\mathrm{Wg}(\nu,m)\sim
m^{-(k+|\nu|)}$, the identity value $\mathrm{Wg}(e,m)=m^{-k}(1+O(k^2/m^2))$, and
the boundedness $|\mathrm{M\ddot ob}(\nu)|\le 4^{|\nu|}$. Under the Gaussian
model, Isserlis' theorem gives the same monomial as a single sum over pairings,
\begin{equation}
    \mathbb E_{\mathrm{Gauss}}\!\Big[\textstyle\prod_a Z_{i_aj_a}\prod_a
    Z^*_{i'_aj'_a}\Big]
    = \sum_{\sigma\in S_k}
      \delta_{\mathbf i,\mathbf i'\sigma}\,\delta_{\mathbf j,\mathbf j'\sigma}.
    \label{eq:appB_wick}
\end{equation}

Writing the index sum weighted by the deterministic operator factors as
\begin{equation}
    W_{\sigma,\tau}
    = \sum_{\text{indices}}\delta_{\mathbf i,\mathbf i'\sigma}\,
      \delta_{\mathbf j,\mathbf j'\tau}\,\Gamma,
    \label{eq:appB_W}
\end{equation}
where $\Gamma$ collects the matrix elements of $M_{vu}$ and $M^\dagger_{vu}$ glued
in the cyclic trace order, the two moments are
\begin{align}
    \mathbb E_{\mathrm{Gauss}}\,\mathrm{Tr}\big[(XX^\dagger)^k\big]
    &= m^{-k}\sum_{\sigma\in S_k} W_{\sigma,\sigma},
    \label{eq:appB_gauss}\\
    \mathbb E_{\mathrm{Haar}}\,\mathrm{Tr}\big[(XX^\dagger)^k\big]
    &= \sum_{\sigma,\tau\in S_k} W_{\sigma,\tau}\,\mathrm{Wg}(\sigma^{-1}\tau,m),
    \label{eq:appB_haar}
\end{align}
the Gaussian rescaling $A_{vu}=Z_{vu}$ versus $A_{vu}=\sqrt m\,U_{vu}$ matching the
$\sigma=\tau$ term of \eqref{eq:appB_haar} weighted by $\mathrm{Wg}(e,m)=m^{-k}$.
Separating that diagonal term,
\begin{align}
    \mathbb E_{\mathrm{Haar}}\,\mathrm{Tr}\big[(XX^\dagger)^k\big]
    &= \big(1+O(k^2/m^2)\big)\,
      \mathbb E_{\mathrm{Gauss}}\,\mathrm{Tr}\big[(XX^\dagger)^k\big]
    + R,
    \nonumber\\
    R &= \!\!\sum_{\substack{\sigma,\tau\\ \sigma\ne\tau}}\!\!
        W_{\sigma,\tau}\,\mathrm{Wg}(\sigma^{-1}\tau,m).
    \label{eq:appB_diag}
\end{align}
The diagonal prefactor $1+O(k^2/m^2)=1+\tilde O(n^2/m^2)$ at $k\sim n\log m$ is
subleading to the off-diagonal correction below; the entire ensemble difference is
the remainder $R$.

\subsection{Genus-resolved bound on the remainder}
\label{appB:genus}

Set $\pi=\sigma^{-1}\tau\ne e$. The magnitude of $W_{\sigma,\tau}$ relative to the
diagonal $W_{\sigma,\sigma}$ is governed by the number of free index loops of the
glued ribbon diagram: each independent loop is an unconstrained index summation
contributing a factor $m$. For the bare gluing of the $U$ indices, the loop count
obeys the Euler relation
\begin{equation}
    L(\sigma,\pi) = L_0(\sigma) - |\pi| + 2g(\sigma,\pi),
    \qquad g\ge 0,\ \ g\le\lfloor|\pi|/2\rfloor,
    \label{eq:appB_euler_bare}
\end{equation}
where $L_0$ is the planar (diagonal) loop count and $g$ the genus of the surface
obtained by gluing the ribbon according to $\sigma$ and $\sigma\pi$; this is the
standard topological expansion of unitary moments~\cite{weingarten1978asymptotic,brouwer1996diagrammatic,collins2003moments,collins2006integration}. The physical moment differs by
the deterministic decorations $\Gamma$---the operators $M_{vu}$ at each vertex and
the projectors $P,Q$ enforcing the occupation constraints. The following
conjecture asserts that these decorations preserve the genus ordering; it is the
one unproven structural input of the appendix, and we discuss its status---and the
tension it is in---immediately after.

\begin{conjecture}[Decorated genus ordering]
\label{conj:transfer}
Let $L_\Gamma(\sigma,\pi)$ be the free-loop count of the decorated ribbon
diagram, including the operators $M_{vu}$ and the projectors $P,Q$. Then
\begin{equation}
    L_\Gamma(\sigma,\pi) \le L_0(\sigma) - |\pi| + 2g(\sigma,\pi),
    \label{eq:appB_decorated}
\end{equation}
so that a contribution at distance $|\pi|=\ell$ and genus $g$ carries a net power
relative to the diagonal of at most $m^{-2(\ell-g)}$.
\end{conjecture}

\paragraph{Status of Conjecture~\ref{conj:transfer}.}
The heuristic mechanism is the following. The operators $M_{vu}$ are local: $M_{vu}$ raises
the occupation at site $v$, lowers it at $u$, and is diagonal in all other site
occupations (Lemmas~\ref{lem:row}--\ref{lem:col}), so it identifies only the pair
$(u,v)$ of half-edges at its vertex and imposes occupation-preserving constraints
elsewhere. The projectors $P,Q$ impose \emph{global} constraints---$Q$ that all
occupied sites are distinct, $P$ that exactly one site is at occupation $d$---which
act as additional $\delta$-identifications among the indices. A constraint of the
form $\delta_{v_A,v_B}$ between two indices that lie in distinct free loops of the
bare gluing \emph{merges} those loops, reducing the loop count by one and removing
a power of $m$; in the cases we check it can never split a loop into two. We verify
\eqref{eq:appB_decorated} explicitly at leading order: for $|\pi|=1$ (a single
transposition, forced to $g=0$) the bare net power is $m^{-2}$, and the
decorations, acting through merges, do not raise it, so the leading off-diagonal
contribution is $O(m^{-2})$; for $|\pi|=2$ the two cases $g\in\{0,1\}$ give bare
net powers $m^{-4}$ and $m^{-2}$, again not raised by the decorations.

We do \emph{not}, however, establish the bound for all $\ell$, and we flag the
reason explicitly rather than assert the general case. The decorations do not enter
the index sum uniformly: the projectors $P,Q$ pin some index loops to the $O(n)$
occupied sites while others range freely over the $O(m)$ empty modes, so the
loop weights are heterogeneous in $n$ and $m$ rather than carrying a uniform power
of $m$. This heterogeneity is the same one visible already at $k=1$ in the variance
asymmetry $\lVert V_{\mathrm{row}}\rVert=O(n^2)$ versus
$\lVert V_{\mathrm{col}}\rVert=O(m)$ of Lemmas~\ref{lem:row}--\ref{lem:col}.
Whether the uniform genus ordering of \eqref{eq:appB_decorated} survives this
heterogeneity for all $\ell$---equivalently, whether the deterministic
Fock-trace factor $\Gamma$ and the Weingarten $\delta$-contractions can be
disentangled so that the net power counting is governed by genus alone---is not
settled by the leading-order check, and a complete analysis must track the $n/m$
scaling of $\Gamma$ separately from the combinatorial $\delta$-structure. We
therefore state \eqref{eq:appB_decorated} as a conjecture and present the remainder
of the appendix as conditional on it.

It is worth making explicit how this ratio is formed, since numerator and
denominator sum over different index sets. Each off-diagonal $W_{\sigma,\tau}$ of
Eq.~\eqref{eq:appB_haar} carries the same deterministic Fock-trace factor
$\Gamma$ as the diagonal term $W_{\sigma,\sigma}$ sharing its $\sigma$; because
$\Gamma$ factorizes over the cyclic trace and is common to both, the outer sum
over $\sigma$ cancels between $R$ and the Gaussian moment
$\mathbb E_{\mathrm{Gauss}}[\cdots]=m^{-k}\sum_{\sigma}W_{\sigma,\sigma}$ of
Eq.~\eqref{eq:appB_gauss}. The ratio $|R|/\mathbb E_{\mathrm{Gauss}}[\cdots]$ is
therefore governed by the inner sum over $\pi=\sigma^{-1}\tau\ne e$ alone, each
term suppressed relative to the diagonal by the loop deficit $m^{-2(\ell-g)}$ of
Eq.~\eqref{eq:appB_decorated} at distance $|\pi|=\ell$ and genus $g$.

Granting Conjecture~\ref{conj:transfer}, we bound $R$. Organize the sum by the
topological index $h=\ell-g\ge 1$, which controls the net suppression
$m^{-2h}$ directly. A contribution of index $h$ has, in the worst case, distance
$\ell=2h$ and genus $g=h$; its support involves at most $2\ell=4h$ of the $k$
active indices. Summing over the admissible support sizes $\ell<s\le2\ell$---the
maximal support $s=2\ell$, corresponding to $\ell$ disjoint transpositions,
dominating, each smaller-support cycle type down by at least a factor $k$ and
contributing an overall constant $C_0=O(1)$ (uniform in $\ell$ for $k\gg\ell$)
absorbed into $C$ below---the number of contributing permutations obeys
\begin{equation}
    \#\{\pi\in S_k:|\pi|=\ell\}\le C_0\binom{k}{2\ell}(2\ell-1)!!
    \;\le\; C_0\,\frac{k^{2\ell}}{2^{\ell}\,\ell!},
    \label{eq:appB_count}
\end{equation}
which for $\ell=2h$ gives at most $C_0\,k^{4h}/(2^{2h}(2h)!)$ configurations. Combining
the count, the net suppression $m^{-2h}$, the bounded Möbius and operator
factors ($|\mathrm{M\ddot ob}|\le 4^{|\pi|}$, $\gamma_{\mathfrak g}=O(1)$) and the
support-sum constant $C_0$, all absorbed
into an absolute constant $C$, and summing over $h$,
\begin{equation}
    \frac{|R|}{\mathbb E_{\mathrm{Gauss}}[\cdots]}
    \;\le\; \sum_{h\ge1}\frac{1}{(2h)!}\Big(\frac{C\,k^{4}}{m^{2}}\Big)^{h}
    \;=\; \cosh\!\Big(\sqrt{\tfrac{C k^4}{m^2}}\Big)-1
    \;=\; O\!\Big(\frac{k^4}{m^2}\Big),
    \label{eq:appB_sum}
\end{equation}
the series converging whenever $k^4/m^2\to0$, i.e.\ $m\gg k^2$. The leading term
is $h=1$: its dominant representative is $\ell=2$, $g=1$, of size $O(k^4/m^2)$,
which exceeds the $\ell=1$, $g=0$ term $O(k^2/m^2)$ once $k\gg1$. The net power is
even, $m^{-2h}$, so no $m^{-1}$ correction arises---the single transposition is
suppressed at $m^{-2}$, not $m^{-1}$.

\subsection{Uniformity in $k$ and conclusion}
\label{appB:conclusion}

At the operating order $k\sim n\log m$, the relative correction \eqref{eq:appB_sum}
is
\begin{equation}
    \frac{|R|}{\mathbb E_{\mathrm{Gauss}}[\cdots]}
    = O\!\Big(\frac{k^4}{m^2}\Big) = \tilde O\!\Big(\frac{n^4}{m^2}\Big),
    \label{eq:appB_rel}
\end{equation}
with convergence condition $m\gg k^2$, i.e.\ $m\gg n^2\log^2 m$, hence
$m=\tilde\Omega(n^2)$. We must check this correction does not accumulate as the
operator norm is extracted. The point is that the bound \eqref{eq:appB_sum} is
controlled by the topological index $h$, whose suppression $m^{-2h}$ is
independent of $k$: increasing $k$ lengthens the trace but does not change the
per-index-$h$ suppression, and the convergent series in $k^4/m^2$ has the same
$O(k^4/m^2)$ leading behaviour for every $k$ in the operating window. Thus
\begin{equation}
    \mathbb E_{\mathrm{Haar}}\,\mathrm{Tr}\big[(XX^\dagger)^k\big]
    = \mathbb E_{\mathrm{Gauss}}\,\mathrm{Tr}\big[(XX^\dagger)^k\big]
      \big(1+\tilde O(n^4/m^2)\big),
    \label{eq:appB_moment_final}
\end{equation}
uniformly in $k$. Substituting into the Markov bound \eqref{eq:appB_markov} at the
optimizing order \eqref{eq:appB_kopt} reproduces, for the Haar ensemble, the
Gaussian threshold $t=\tilde O(\sqrt n)$ of Section~\ref{subsec:concentration}
multiplied by $(1+\tilde O(n^4/m^2))^{1/2k}\to 1$, so taking the spectral norm,
\begin{equation}
    \lVert \mathcal A_{\mathrm{Haar}}\rVert
    \le \lVert \mathcal A_{\mathrm{Gauss}}\rVert\,\big(1+\tilde O(n^4/m^2)\big),
    \label{eq:appB_final}
\end{equation}
which is Lemma~\ref{lem:weingarten}, conditional on
Conjecture~\ref{conj:transfer}. In the operating regime $m=\tilde\Omega(n^3)$
the correction is $\tilde O(n^4/m^2)=\tilde O(n^{-2})\to0$, comfortably inside the
convergence threshold $m=\tilde\Omega(n^2)$ that the transfer itself requires.
Thus, \emph{if} Conjecture~\ref{conj:transfer} holds, the Haar-to-Gaussian
transfer is not the binding constraint on the mode-scaling law: the
$m=\tilde\Omega(n^3)$ requirement is set by the leakage--collision product of
Section~\ref{subsec:scaling}, and the transfer would then be satisfied with room
to spare. Establishing the conjecture for all $\ell$, and thereby making the Haar
statement unconditional, is left as the natural next step; the leading-order
verification above and the variance asymmetry of
Lemmas~\ref{lem:row}--\ref{lem:col} indicate both why the bound is plausible and
where the difficulty lies.

\section{Numerical Grid and Validation}
\label{app:grid}

This appendix records the complete parameter grid underlying
Fig.~\ref{fig:numerics}, the exact sector dimensions at each point, and the
validation of the sparse leakage-operator construction. All points use the
per-point deterministic seeding and master seed described in
Sec.~\ref{subsec:numerics}. Table~\ref{tab:validation} reports the agreement of
the sparse norm with a brute-force full-Fock-space computation;
Table~\ref{tab:grid} lists every $(d,n,m)$ measured, its collision-free and
single-bunch sector dimensions, the Paulsen-dilated operator dimension $D$, and
whether the point is dilute ($m\ge 2n^2$). The largest operator assembled has
$D\approx3.7\times10^{7}$ (at $d=3$, $n=8$, $m=30$); a dilute $d=2$ point at
large $n$ would instead carry $D\gtrsim10^{11}$ and is correspondingly out of
reach, which is why the large-$n$ $n$-scan points are taken sub-dilute and drawn
as open markers in Fig.~\ref{fig:numerics}.

\begin{table}
\caption{Validation of the sparse construction against a brute-force
diagonalization of the full fixed-$N$ Fock-space Hamiltonian projected onto the
$Q$ and $P$ sectors, at six small $(d,n,m)$ points. The sparse leakage-operator
norm $\lVert X\rVert$ and the brute-force $\lVert QH_{\mathrm{Lie}}P\rVert$ agree
to ten significant figures at every point.}
\label{tab:validation}
\begin{ruledtabular}
\begin{tabular}{cccccc}
$d$ & $n$ & $m$ & sparse $\lVert X\rVert$ & brute force & $|\Delta|$ \\
\colrule
2 & 3 & 4 & 1.8108197284 & 1.8108197284 & $<10^{-10}$ \\
2 & 4 & 5 & 1.8726488072 & 1.8726488072 & $<10^{-10}$ \\
3 & 5 & 5 & 2.7493249439 & 2.7493249439 & $<10^{-10}$ \\
3 & 4 & 6 & 2.6113134026 & 2.6113134026 & $<10^{-10}$ \\
2 & 5 & 6 & 3.1699708798 & 3.1699708798 & $<10^{-10}$ \\
3 & 6 & 5 & 3.7914826009 & 3.7914826009 & $<10^{-10}$ \\
\end{tabular}
\end{ruledtabular}
\end{table}

\begin{table*}
\caption{Complete numerical grid (QUICK-mode run, master seed \texttt{20260711},
five realizations per point for the four-ensemble points and three for the
Haar-only points). For each $(d,n,m)$ the table lists the coupling ensembles
measured (``all 4'' $=$ \{gauss iid, gauss Herm, Haar, Herm.\ $i\log U$\}; ``Haar''
$=$ Haar entries only), the collision-free sector dimension
$\dim\mathcal{H}_Q=\#\{n_k\le d-1,\ \sum_k n_k=n\}$, the single-bunch sector
dimension $\dim\mathcal{H}_P=m\cdot\#\{n_k\le d-1\ (k\ne u_0),\ \sum n_k=n-d\}$,
the leakage-operator (Paulsen-dilated) dimension $D=\dim\mathcal{H}_Q+\dim\mathcal{H}_P$,
and the dilution status (``yes'' if $m\ge 2n^2$, ``sub'' if sub-dilute). The
assembled sparse operator uses the reached portion of $\mathcal{H}_Q$, a subset of
the quoted $\dim\mathcal{H}_Q$ that leaves the norm unchanged.}
\label{tab:grid}
\begin{ruledtabular}
\begin{tabular}{cccccccc}
$d$ & $n$ & $m$ & ensembles & $\dim\mathcal{H}_Q$ & $\dim\mathcal{H}_P$ & $D$ & dilute? \\
\colrule
2 & 3 & 12 & all 4 & 220 & 132 & 352 & sub \\
2 & 3 & 24 & all 4 & 2,024 & 552 & 2,576 & yes \\
2 & 3 & 48 & all 4 & 17,296 & 2,256 & 19,552 & yes \\
2 & 3 & 96 & all 4 & 142,880 & 9,120 & 152,000 & yes \\
2 & 3 & 192 & all 4 & 1,161,280 & 36,672 & 1,197,952 & yes \\
2 & 3 & 384 & Haar & 9,363,584 & 147,072 & 9,510,656 & yes \\
2 & 4 & 16 & all 4 & 1,820 & 1,680 & 3,500 & sub \\
2 & 4 & 24 & all 4 & 10,626 & 6,072 & 16,698 & sub \\
2 & 4 & 32 & all 4 & 35,960 & 14,880 & 50,840 & yes \\
2 & 4 & 48 & all 4 & 194,580 & 51,888 & 246,468 & yes \\
2 & 4 & 64 & all 4 & 635,376 & 124,992 & 760,368 & yes \\
2 & 4 & 96 & Haar & 3,321,960 & 428,640 & 3,750,600 & yes \\
2 & 4 & 128 & Haar & 10,668,000 & 1,024,128 & 11,692,128 & yes \\
2 & 5 & 20 & all 4 & 15,504 & 19,380 & 34,884 & sub \\
2 & 5 & 30 & all 4 & 142,506 & 109,620 & 252,126 & sub \\
2 & 5 & 40 & all 4 & 658,008 & 365,560 & 1,023,568 & sub \\
2 & 5 & 50 & Haar & 2,118,760 & 921,200 & 3,039,960 & yes \\
2 & 5 & 60 & Haar & 5,461,512 & 1,950,540 & 7,412,052 & yes \\
2 & 6 & 24 & all 4 & 134,596 & 212,520 & 347,116 & sub \\
2 & 6 & 30 & Haar & 593,775 & 712,530 & 1,306,305 & sub \\
2 & 6 & 40 & Haar & 3,838,380 & 3,290,040 & 7,128,420 & sub \\
2 & 7 & 20 & all 4 & 77,520 & 232,560 & 310,080 & sub \\
2 & 8 & 18 & all 4 & 43,758 & 222,768 & 266,526 & sub \\
\colrule
3 & 4 & 40 & all 4 & 121,810 & 1,560 & 123,370 & yes \\
3 & 4 & 60 & Haar & 592,065 & 3,540 & 595,605 & yes \\
3 & 5 & 30 & all 4 & 264,306 & 13,050 & 277,356 & sub \\
3 & 5 & 60 & Haar & 7,514,712 & 106,200 & 7,620,912 & yes \\
3 & 6 & 24 & all 4 & 412,896 & 54,648 & 467,544 & sub \\
3 & 6 & 48 & Haar & 22,017,808 & 882,096 & 22,899,904 & sub \\
3 & 7 & 20 & all 4 & 484,500 & 139,080 & 623,580 & sub \\
3 & 7 & 36 & Haar & 24,039,972 & 2,613,240 & 26,653,212 & sub \\
3 & 8 & 18 & all 4 & 633,726 & 319,464 & 953,190 & sub \\
3 & 8 & 30 & Haar & 30,462,615 & 6,741,630 & 37,204,245 & sub \\
\colrule
4 & 5 & 24 & Haar & 97,704 & 552 & 98,256 & sub \\
4 & 6 & 20 & Haar & 172,900 & 3,800 & 176,700 & sub \\
4 & 7 & 18 & Haar & 325,584 & 17,442 & 343,026 & sub \\
4 & 8 & 16 & Haar & 428,418 & 48,720 & 477,138 & sub \\
\colrule
5 & 6 & 20 & Haar & 176,700 & 380 & 177,080 & sub \\
5 & 7 & 18 & Haar & 343,026 & 2,754 & 345,780 & sub \\
5 & 8 & 16 & Haar & 477,258 & 10,880 & 488,138 & sub \\
\end{tabular}
\end{ruledtabular}
\end{table*}

\bibliography{sample.bib}

\end{document}